\documentclass{article}

\usepackage{microtype}
\usepackage{graphicx}
\usepackage{subfigure}
\usepackage{booktabs} 

\usepackage{hyperref}

\usepackage[accepted]{icml2023}

\usepackage{amsmath}
\usepackage{amssymb}
\usepackage{mathtools}
\usepackage{amsthm}

\usepackage[capitalize,noabbrev]{cleveref}

\theoremstyle{plain}
\newtheorem{theorem}{Theorem}[section]

\newtheorem{lemma}[theorem]{Lemma}
\newtheorem{corollary}[theorem]{Corollary}
\theoremstyle{definition}
\newtheorem{definition}[theorem]{Definition}
\newtheorem{assumption}[theorem]{Assumption}
\theoremstyle{remark}

\usepackage[textsize=tiny]{todonotes}

\let\oldemptyset\emptyset
\DeclareMathOperator{\R}{\mathbb{R}}
\DeclareMathOperator{\E}{\mathbb{E}}

\newcommand{\Tau}{\mathcal{T}}

\usepackage{bbm, dsfont}

\icmltitlerunning{Context-Aware Bayesian Network Actor-Critic Methods for Cooperative Multi-Agent Reinforcement Learning}

\begin{document}

\twocolumn[
\icmltitle{Context-Aware Bayesian Network Actor-Critic Methods for Cooperative Multi-Agent Reinforcement Learning}



\icmlsetsymbol{equal}{*}

\begin{icmlauthorlist}
\icmlauthor{Dingyang Chen}{yyy}
\icmlauthor{Qi Zhang}{yyy}
\end{icmlauthorlist}

\icmlaffiliation{yyy}{Artificial Intelligence Institute, University of South Carolina, SC, USA}

\icmlcorrespondingauthor{Dingyang Chen}{dingyang@email.sc.edu}
\icmlcorrespondingauthor{Qi Zhang}{qz5@cse.sc.edu}

\icmlkeywords{Machine Learning, ICML}

\vskip 0.3in
]



\printAffiliationsAndNotice{} 

\begin{abstract}
Executing actions in a correlated manner is a common strategy for human coordination that often leads to better cooperation, which is also potentially beneficial for cooperative multi-agent reinforcement learning (MARL). However, the recent success of MARL relies heavily on the convenient paradigm of purely decentralized execution, where there is no action correlation among agents for scalability considerations. In this work, we introduce a Bayesian network to inaugurate correlations between agents' action selections in their joint policy. Theoretically, we establish a theoretical justification for why action dependencies are beneficial by deriving the multi-agent policy gradient formula under such a Bayesian network joint policy and proving its global convergence to Nash equilibria under tabular softmax policy parameterization in cooperative Markov games. Further, by equipping existing MARL algorithms with a recent method of differentiable directed acyclic graphs (DAGs), we develop practical algorithms to learn the context-aware Bayesian network policies in scenarios with partial observability and various difficulty. We also dynamically decrease the sparsity of the learned DAG throughout the training process, which leads to weakly or even purely independent policies for decentralized execution. Empirical results on a range of MARL benchmarks show the benefits of our approach. The code is available at https://github.com/dchen48/BNPG.
\end{abstract}

\section{Introduction}
\label{submission}

Cooperative multi-agent reinforcement learning (MARL) methods equip a group of autonomous agents with the capability of planning and learning to maximize their joint utility, or reward signals in the reinforcement learning (RL) literature, which provides a promising paradigm for a range of real-world applications, such as traffic control  \cite{chu2019multi}, coordination of multi-robot systems  \cite{corke2005networked}, and power grid management  \cite{callaway2010achieving}.
As a key distinction from the single-agent setting, multi-agent joint action spaces grow exponentially with the number of agents, which imposes significant scalability issues.
As a convenient and commonly adopted solution, most existing cooperative MARL methods only consider {\em product policies}, i.e., each agent selects its local action independently given the state or its observations.
Restricting to product policies, however, does come at a cost for cooperative tasks:
consider an example where cars wait at a crossroads, it would be hard for the cars to coordinate their movements without knowing others' intentions, potentially resulting in a crash or congestion.
Intuitively, optimizing over the smaller joint policy space of all product policies can lead to suboptimal joint policies compared to optimizing over the entire set of joint policies that also includes {\em correlated policies} where the local actions of all agents are sampled together in a potentially correlated manner.   

The research question then arises naturally:
how can we introduce correlations for cooperative multi-agent joint policies, while taming the scalability issues?
Noting that a joint policy is joint distributions (over agents' local actions), a straightforward yet underexplored solution idea is to use a Bayesian network (BN) that represents conditional dependencies between agents' local actions via a directed acyclic graph (DAG), where a desirable DAG topology structure captures important dependencies that exist among hopefully a set of sparsely connected agents.
As our first contribution, we formalize this solution idea of BN joint policies in the cooperative Markov game framework  \cite{Boutilier1999SequentialOA,peshkin2001learning}, derive its associated BN policy gradient formula, and then prove the global convergence of its gradient ascent to Nash equilibria under the tabular policy parameterization.

As our second contribution, we then adapt existing multi-agent actor-critic methods such as MAPPO  \cite{yu2021surprising} to incorporate BN joint policies.
For practicality and efficiency, our algorithm features the following two key design choices:
(i) 
Our DAG topology of the BN joint policy is learnable to be context-aware based on the environment state or the agents' joint observations, leveraging a recently developed technique for differentiable DAG learning.
(ii)
To execute a BN joint policy, the agents need to communicate their intended actions to their children in the BN, unless the BN's DAG topology reduces to product policies, and the corresponding communication overhead is directly determined by the DAG's denseness/sparseness.
To encourage sparse communication during execution, we develop a learning strategy that dynamically increases the sparsity of the learned DAG, where full sparsity (i.e., product policies) can be achieved at the last stage of the training process and therefore the learned joint policy can be executed in a purely decentralized manner, making our algorithm compatible with the centralized training, decentralized execution (CTDE) paradigm  \cite{lowe2017multi}.
Empirically results show the benefits of our algorithm equipped with these two design choices.

The rest of this paper is structured as follows:
Section \ref{sec:Releated Work} reviews closely related work; 
Section \ref{sec:Preliminaries} introduces preliminaries of cooperative Markov games and solution concepts therein;
Section \ref{sec:Bayesian network joint policy} formulates our novel notion of Bayesian network joint policy, followed by the theoretical results in Section \ref{sec:Asymptotic convergence of the policy gradient dynamics};
Section \ref{sec:Practical Algorithm} describes our practical algorithm, followed by the empirical results in Section \ref{sec:Experiments};
Section \ref{sec:Conclusion} concludes the paper.

\section{Related Work}
\label{sec:Releated Work}
\paragraph{Convergence of policy gradient in cooperative MGs.}
Cooperative MGs are an important subclass of Markov games, where each agent has the same reward function. Recent work has established the convergence guarantee of policy gradient in Cooperative MGs to Nash policies under tabular setting with direct parameterization \cite{leonardos2021global}, and with softmax parameterization \cite{zhang2022effect,chen2022convergence}. 
\paragraph{Policy correlations in MARL.} Some prior work has noticed the limitation of purely decentralized execution and made some efforts to introduce correlations among policies. Value-based method  \cite{rashid2018qmix} following the CTDE-based training paradigm has been combined with coordination graph  \cite{bohmer2020deep} for introducing pairwise correlation. However, the optimization process requires Max-Sum  \cite{Rogers2009BoundedAD} which is computationally intensive when the coordination graph is dense.  \cite{wang2022contextaware} proposes a rule-based pruning method to generate a sparse coordination graph that can speed up the Max-Sum algorithm without harming the performance. However, the extension from the pairwise correlation to more complicated ones is not trivial. There are also some policy-based algorithms augmented with correlated execution.  
\cite{ruan2022gcs} combines MAPPO \cite{yu2021surprising} with a graph generator outputting Bayesian Network that determines action dependencies. The optimization of the graph generator is achieved by maximizing the cumulative rewards constrained to the depth and dagness of the output graph. However, there is no theoretical justification for why using Bayesian networks is reasonable, and the output graph is not guaranteed to be a DAG, which requires some non-differentiable rule-based pruning and can harm performance. Moreover, the existing methods do not generate fully decentralized policies at the end of the training, which increases the execution time in the deployment of the model.   

\paragraph{Differentiable DAG learning.}
The goal is to learn such an adjacency matrix that can help the actors better coordinate. However, the generation of DAG requires non-differentiable operations due to its discreteness and acyclicity, which prevents end-to-end training. Fortunately, recent work  \cite{charpentier2022differentiable} proposes a simple fully differentiable DAG learning algorithm. Every DAG can be decomposed into the multiplication of a permutation matrix determining topological ordering and upper triangular matrix (edge matrix) determining DAG structure. We can use neural networks to learn the logits for permutation and edge matrices, and use  Gumbel-Softmax  \cite{jang2016categorical} and Gumbel-Sinkhorn  \cite{mena2018learning} to differentiably transform them to the corresponding discrete ones. 

\section{Preliminaries}
\label{sec:Preliminaries}
\paragraph{Cooperative Markov game.}
We consider a cooperative Markov game (MG) $\langle \mathcal{N},\mathcal{S},\mathcal{A}, P, r, \mu \rangle$ with
$N$ agents indexed by $i\in\mathcal{N}=\{1,...,N\}$,
state space $\mathcal{S}$,
action space $\mathcal{A} = \mathcal{A}^1\times\cdots\times\mathcal{A}^N$,
transition function $P: \mathcal{S}\times\mathcal{A}\to\Delta(\mathcal{S})$, (team) reward function $r: \mathcal{S}\times\mathcal{A}\to\R$ shared by all agents $i\in\mathcal{N}$,
and initial state distribution $\mu \in \Delta(\mathcal{S})$,
where we use $\Delta(\mathcal{X})$ to denote the set of probability distributions over $\mathcal{X}$.
For ease of exposition, we assume full observability, i.e., each agent observes the global state $s\in\mathcal{S}$, until Section \ref{sec:Practical Algorithm} where we introduce our practical algorithm that incorporates partial observability.
Under full observability, we consider general {\em joint policy}, $\pi:\mathcal{S}\to\Delta(\mathcal{A})$, which maps from the state space to distributions over the joint action space.
As the size of action space $\mathcal{A}$ grows exponentially with $N$, the commonly used joint policy subclass is the {\em product policy},  $\pi=(\pi^1,\cdots,\pi^N):\mathcal{S}\to\times_{i\in\mathcal{N}}\Delta(\mathcal{A}^i)$, which is factored as the product of local policies $\pi^i:\mathcal{S}\to\Delta(\mathcal{A}^i)$, $\pi(a|s) = \prod_{i\in\mathcal{N}}\pi^i(a^i|s)$, each mapping the state space only to the action space of an individual agent.   
Define the discounted return from time step $t$ as $G_t = \sum_{l=0}^{\infty} \gamma^l r_{t+l}$, where $r_t:=r(s_t,a_t)$ is the team reward at time step $t$. Joint policy $\pi$ induces a value function defined as $V_\pi(s_t) = \E_{s_{t+1:\infty}, a_{t:\infty}\sim\pi}[G_t|s_t]$, and action-value function $Q_\pi(s_t, a_t) = \E_{s_{t+1:\infty}, a_{t+1:\infty}\sim\pi}[G_t|s_t, a_t]$.
Following policy $\pi$, the cumulative team reward, i.e., the value function, starting from $s_0\sim\mu$ is denoted as $V_\pi(\mu):=\E_{s_0\sim\mu}[V_\pi(s_0)]$. The (unnormalized) {\em discounted state visitation measure} by following policy $\pi$ after starting at $s_0\sim\mu$ is defined as 
\begin{align*}
d^\pi_\mu(s):=\E_{s_0\sim\mu}\left[\sum_{t=0}^{\infty}\gamma^t{\rm Pr}^\pi(s_t=s|s_0)\right]
\end{align*}
where ${\rm Pr}^\pi(s_t=s|s_0)$ is the probability that $s_t=s$ when starting at state $s_0$ and following $\pi$ subsequently.

As all agents share a team reward, cooperative MARL considers the same objective as single-agent RL of optimizing the joint policy from experience to maximize its value, i.e., $\max_\pi V_\pi(\mu)$. 
For product policies, we will also consider the weaker solution concept of the Nash policy, as formally defined below.
\begin{definition}[Nash policy]
\label{Nash policy}
Product policy $\pi=(\pi^1,\cdots,\pi^N)=(\pi^i,\pi^{-i})$ is a Nash policy if  
$$
    \forall i\in \mathcal{N}, \forall \bar{\pi}^{i}\in \Delta(\mathcal{A}^i),  V_{\bar{\pi}^{i},\pi^{-i}}(\mu) \leq V_{\pi}(\mu)
$$
where $\pi^{-i}$ is the local policies of the agents excluding $i$.
\end{definition}
For a Nash policy, each agent $i$ maximizes the value function given fixed local policies of other agents.
\begin{figure}[t]
\begin{center}
\centerline{\includegraphics[width=.97\columnwidth]{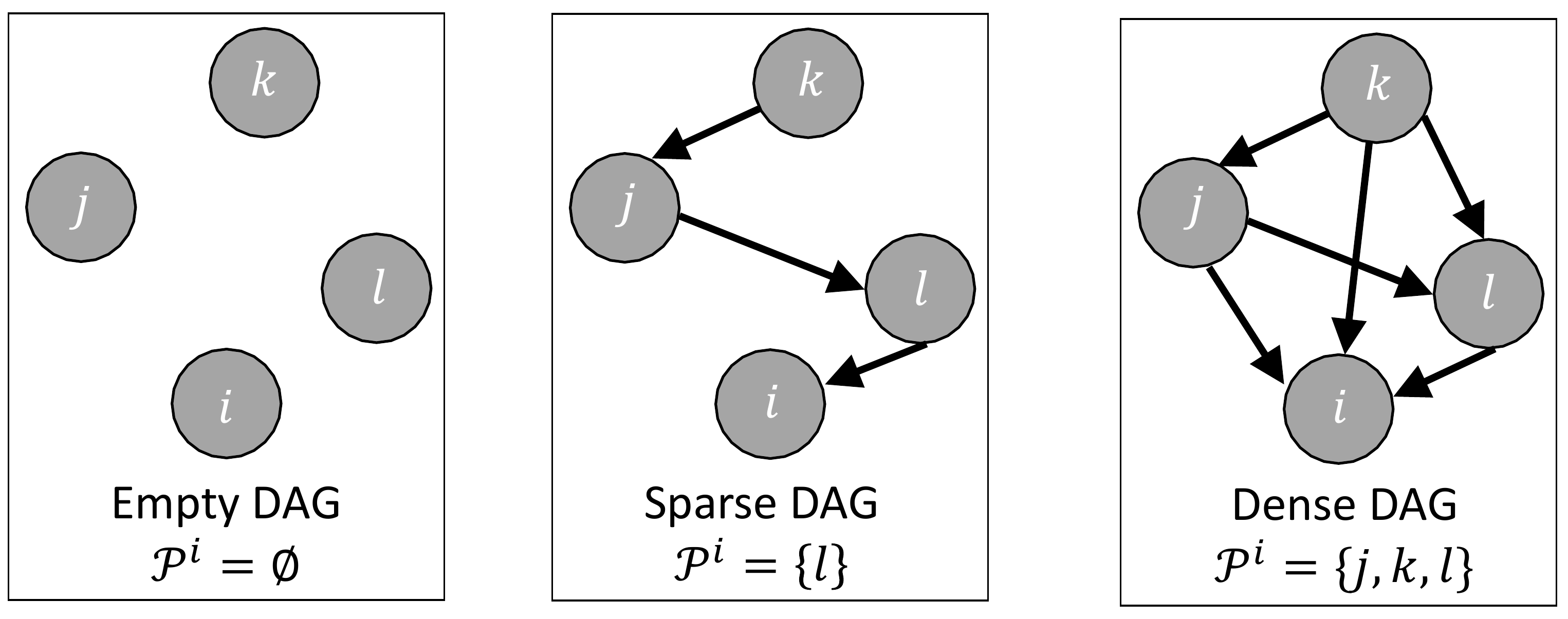}}
\caption{Illustration of various DAG topologies.}
\label{fig:Bayesian network illustration}
\end{center}
\vskip -0.3in
\end{figure}

\section{Bayesian Network Joint Policy}
\label{sec:Bayesian network joint policy}
Most existing cooperative MARL methods consider only product policies to optimize, rather than the more general set of general joint policies.
This is mainly because product policies can conveniently deal with the scalability issue of the joint action space.
Another justification is that restricting to product policies incurs no optimality gap, since it is well-known that there is always an optimal joint policy that is deterministic and therefore a product policy.
However, the existence of an optimal product policy does not guarantee that we can search it out easily.
In fact, existing theoretical and empirical results \cite{leonardos2021global,zhang2022effect,chen2022convergence}, including ours in this paper, have shown that restrictively searching the product policies via gradient ascent can only find local optima such as Nash policies, even in the noiseless tabular setting.

As the key notion in this work, we now formally introduce a class of joint policies that is more general than product policies by introducing action dependencies captured by a Bayesian network (BN). 
We specify a BN by a directed acyclic graph (DAG) $\mathcal{G}=(\mathcal{N}, \mathcal{E})$ with vertex set $\mathcal{N}$ and directed edge set $\mathcal{E}\subseteq \{(i,j): i,j\in\mathcal{N}, i\neq j\}$. Denote the parents of agent $i$ as $\mathcal{P}^i:=\{j: (i, j)\in\mathcal{E}\}$, and the corresponding parent actions as $a^{\mathcal{P}^i}\in\times_{j\in\mathcal{P}^i}\mathcal{A}^j$, as illustrated in Figure \ref{fig:Bayesian network illustration}. Under full observability and with BN $\mathcal{G}$, we consider a {\em BN (joint) policy}, $(\pi,\mathcal{G})=(\pi^1, \cdots, \pi^N,\mathcal{G}):\mathcal{S}\to\Delta(\mathcal{A})$. Similar to product policies, BN policies are also factored as the product of local policies given the state and action dependencies determined by $\mathcal{G}$, i.e., $\pi^i:\mathcal{S}\times (\times_{j\in\mathcal{P}^i}\mathcal{A}^j) \to\Delta(\mathcal{A}^i)$, and thus joint action $a=(a^1,\cdots,a^N)$ is sampled as $\pi(a|s) = \prod_{i\in\mathcal{N}}\pi^i(a^i|s,a^{\mathcal{P}^i})$. 

We make two remarks on BN policies:
(i) With the factorization of a BN policy into its local policies, the concept of Nash policy in Definition \ref{Nash policy} is also applicable to BN policies.
(ii) BN policies naturally interpolate product policies and general joint policies, including them as two extremes:
BN policies reduce to product policies when DAG $\mathcal{G}$ is an empty graph (Figure \ref{fig:Bayesian network illustration} (left)) and can model general joint policies when  $\mathcal{G}$ is dense (Figure \ref{fig:Bayesian network illustration} (right)).

\section{Convergence of the Tabular Softmax BN Policy Gradient in Cooperative MGs}
\label{sec:Convergence of the tabular softmax policy gradient in MPGs}
In this section, we consider optimizing BN policies through policy gradient ascent under the tabular softmax parameterization.
Under the same assumptions, we are able to extend existing convergence results from products policies to BN policies, asserting that optimizing BN policies through gradient ascent can indeed find global optima (rather than Nash) when the BN's DAG is dense.

Formally, the local policies in the BN policy are parameterized in the tabular softmax manner from the global state and parent actions, i.e., we have, for each agent $i$, its policy parameter 
\begin{align*}
\theta^i=\left\{\theta_{s,a^{\mathcal{P}^i},a^i}^i\in\mathbb{R}:s\in \mathcal{S}, a^{\mathcal{P}^i}\in\times_{j\in\mathcal{P}^i}\mathcal{A}^j,a^i\in \mathcal{A}^i\right\}
\end{align*} 
and induced softmax local policy 
\begin{align}
\label{eq:tabular softmax local policy}
\pi^i_{\theta^i}\left(a^i|s,a^{\mathcal{P}^i}\right)\propto\exp\left(\theta^i_{s,a^{\mathcal{P}^i},a^i}\right)
\end{align}
with the BN policy parameterized as $\pi_\theta = (\pi^1_{\theta^1},\cdots,\pi^N_{\theta^N})$.

In Lemma \ref{lemma:state-based tabular softmax multi-agent policy gradient}, we derive the policy gradient form for the BN policy as parameterized in Equation \eqref{eq:tabular softmax local policy}, which will used to establish our convergence results in this section.

It will be also convenient to introduce a few shorthands before stating Lemma \ref{lemma:state-based tabular softmax multi-agent policy gradient}.
Consider a subset $\mathcal{M}\subseteq\mathcal{N}$ of all agents and its complement $-\mathcal{M}$, such that a joint action can be decomposed as $a=(a^{\mathcal{M}},a^{-\mathcal{M}})$.
Let 
\begin{align*}
    \pi^{\mathcal{M}}(a^{\mathcal{M}}|s,a^{-\mathcal{M}}):= 
    \frac{\pi(a^{\mathcal{M}},a^{-\mathcal{M}}|s)}{\sum_{\bar{a}^{\mathcal{M}}}\pi(\bar{a}^{\mathcal{M}},a^{-\mathcal{M}}|s)}
\end{align*}
be the conditional for $a^{\mathcal{M}}$ under $\pi$.
Let 
\begin{align*}
Q_\pi(s,a^{\mathcal{M}}) := \E_{a^{-\mathcal{M}}\sim\pi^{-\mathcal{M}}(\cdot|s,a^{\mathcal{M}})}
\left[Q_\pi(s,a^{\mathcal{M}},a^{-\mathcal{M}})\right].
\end{align*}
Let $\mathcal{P}^i_+ := \mathcal{P}^i \cup \{i\}$ denote the set of agent $i$ and its parents.
We will also abbreviate $V_{\pi_\theta}$, $Q_{\pi_\theta}$ as $V_{\theta}$, $Q_{\theta}$, respectively.

\begin{lemma}[Tabular softmax BN policy gradient form, proof in Appendix \ref{BN_PG}]
\label{lemma:state-based tabular softmax multi-agent policy gradient}
For the tabular softmax BN policy parameterized as in Equation \eqref{eq:tabular softmax local policy}, we have:
\begin{align}
\frac{\partial V_\theta(\mu)}{\partial \theta^i_{s,a^{\mathcal{P}^i},a^i}} =\frac{1}{1-\gamma}d_\mu^{\pi_\theta}(s,a^{\mathcal{P}^i})\pi_{\theta^{i}}^i(a^i|s,a^{\mathcal{P}^i})A^{i}_\theta(s,a^{\mathcal{P}^i},a^i)\nonumber
\end{align}
where
$d_\mu^{\pi_\theta}(s,a^{\mathcal{P}^i}) := d_\mu^{\pi_\theta}(s)\textstyle\sum_{a^{-\mathcal{P}^i}}\pi_{\theta}(a^{-\mathcal{P}^i},a^{\mathcal{P}^i}|s)$,
$A_\theta^{i}(s,a^{\mathcal{P}^i},a^i) :=Q_\theta (s,a^{\mathcal{P}^i_+})-Q_\theta(s,a^{\mathcal{P}^i})$.
\end{lemma}

The policy gradient form in Lemma \ref{lemma:state-based tabular softmax multi-agent policy gradient} generalizes its counterpart for single-agent policies \cite{agarwal2021theory} and for multi-agent product policies \cite{zhang2022effect,chen2022convergence} under the tabular softmax policy parameterization, which enables us to extend the convergence results to the BN joint policies. 

Below we state the assumptions that have been used \cite{zhang2022effect,chen2022convergence}, to generate the convergence results for product policies, i.e., $\mathcal{G}=(\mathcal{N},\emptyset)$.
\begin{assumption}
\label{assumption:discounted state visitation distribution}
For any $\pi$ and any state $s$ of the Markov game, $d^\pi_\mu(s) > 0$.
\end{assumption}

\begin{assumption} [Reward function is bounded]
\label{assumption:Reward function is bounded}
The reward function $r$ is bounded in the range $[r_{\rm min},r_{\rm max}]$, such that the value function $V$ is bounded as  $V_{\rm min}\leq V_\pi(s)\leq V_{\rm max}~\forall s, \pi$.
\end{assumption}

\begin{assumption} 
\label{assumption:PG}
Following the policy gradient dynamics \eqref{eq:PG}, the policy of every agent $i$ converges asymptotically, i.e., $\pi_{\theta^i_t}^i\to\pi_{\theta^i_*}^i$ as $t\to\infty,~\forall i$.
\end{assumption}

Assumption 1 and assumption 2 are standard assumptions used in \cite{agarwal2021theory,zhang2022effect,chen2022convergence}, which ensures the sufficient coverage of all states and the boundness of the reward function, respectively. Assumption 3 is a stronger assumption used in \cite{zhang2022effect,chen2022convergence}. A sufficient condition for assumption \ref{assumption:PG} by \cite{fox2022independent} is that the fixed point of the equation in Lemma \ref{lemma:state-based tabular softmax multi-agent policy gradient} are isolated. The purpose of assumption \ref{assumption:PG} is to establish the convergence of  $A^{i}_\theta(s,a^{\mathcal{P}^i_+})$ if $d_\mu^{\pi_\theta}(s,a^{\mathcal{P}^i})$ is positive. Otherwise, it can be the case that both $\pi_{\theta^{i}}^i(a^i|s,a^{\mathcal{P}^i})$ and $A^{i}_\theta(s,a^{\mathcal{P}^i_+})$ are divergent when the gradient converges to zero.

We next present our convergence results for the standard policy gradient dynamics in Sections \ref{sec:Asymptotic convergence of the policy gradient dynamics}, where Assumptions \ref{assumption:discounted state visitation distribution} - \ref{assumption:PG} hold, with proofs in the appendix \ref{Proof of Theorem}.

\subsection{Asymptotic Convergence of the Tabular Softmax BN Policy Gradient Dynamics}
\label{sec:Asymptotic convergence of the policy gradient dynamics}
In Theorem \ref{theorem:Asymptotic convergence to Nash with gradient ascent}, we establish, under the tabular softmax BN policy parameterization, the asymptotic convergence to a Nash policy in a MPG of the standard policy gradient dynamics:
\begin{align}\label{eq:PG}
    \theta^i_{t+1} = \theta^i_{t} + \eta\nabla_{\theta^i} V_{\theta_{t}}(\mu) 
\end{align}
where $\eta$ is the fixed stepsize and the update is performed by every agent $i\in\mathcal{N}$. 

For each agent $i$, parent actions $a^{\mathcal{P}^i}$, and local action $a^i$, Equation \eqref{eq:PG} becomes 
\begin{align}\label{eq:BayesianPG}
    \theta^{i,t+1}_{s,a^{\mathcal{P}^i},a^i} &= \theta^{i,t}_{s,a^{\mathcal{P}^i},a^i} + \eta\nabla_{\theta^i_{s,a^{\mathcal{P}^i},a^i}} V^i_{\theta_{t}}(\mu)
\end{align}

\begin{theorem}[Asymptotic convergence of BN policy gradient, proof in Appendix \ref{proof:Nash}]
\label{theorem:Asymptotic convergence to Nash with gradient ascent}
Under Assumptions \ref{assumption:discounted state visitation distribution} - \ref{assumption:PG},
suppose every agent $i$ follows the policy gradient dynamics \eqref{eq:PG}, which results in the update dynamics \eqref{eq:BayesianPG} for each each agent $i$, parent actions $a^{\mathcal{P}^i}$, and local action $a^i$, with $\eta\leq \frac{(1-\gamma)^3}{8N(r_{\rm max}-r_{\rm min})}$, then the converged BN policy 
$(\pi_{\theta^1_*}^1, \cdots, \pi_{\theta^N_*}^N,\mathcal{G})$ is a Nash policy.
\end{theorem}

The main trick of our proof for Theorem \ref{theorem:Asymptotic convergence to Nash with gradient ascent} is to view the parent actions $a^{\mathcal{P}^i}$ as part of the state, i.e., $d_\mu^{\pi_\theta}(s,a^{\mathcal{P}^i})$ becomes the new state visitation measure for the augmented state $(s, a^{\mathcal{P}^i})$. After this transformation, the update dynamics in \ref{lemma:state-based tabular softmax multi-agent policy gradient} is resemble to the ones for the product policy,i.e., $\mathcal{G}:=(\mathcal{N}, \emptyset)$, and thus straightforwardly generalize their results for the product joint policy to the BN policy. However, the problem with this formulation of new state $(s, a^{\mathcal{P}^i})$ is that $d_\mu^{\pi_\theta}(s,a^{\mathcal{P}^i}) = d_\mu^{\pi_\theta}(s)\textstyle\sum_{a^{-\mathcal{P}^i}}\pi_{\theta}(a^{-\mathcal{P}^i},a^{\mathcal{P}^i}|s)$ can be zero even if the state visitation measure $d_\mu^{\pi_\theta}(s)$ is strictly positive. This is the main reason we cannot establish results stronger than the ones obtained in  \cite{zhang2022effect}, even for the fully connected Bayesian network with $N(N-1)/2$ edges which intuitively behave similar to the single-agent setting  \cite{agarwal2021theory} and should therefore result in the optimal policy than only a Nash policy.  

\begin{assumption} 
\label{assumption:augumented_state_visitation}
Any augmented state $(s, a^{\mathcal{P}^i})$ has positive visitation measure, i.e., $d_\mu^{\pi_\theta}(s,a^{\mathcal{P}^i})>0$. 
\end{assumption}

\begin{definition}[Fully-correlated BN policy]
\label{definition:epsilon-Nash policy}
A BN policy $(\pi,(\mathcal{N},\mathcal{E}))$ is fully-correlated if $\lvert\mathcal{E}\rvert=N(N-1)/2$, the maximum number of edges in a DAG.
\end{definition}

\begin{corollary}[Asymptotic convergence of BN policy gradient to optimal fully-correlated BN joint policy, proof in Appendix \ref{proof to corollary 1}]
\label{theorem:Asymptotic convergence to optimal with gradient ascent}
Under Assumptions \ref{assumption:discounted state visitation distribution} - \ref{assumption:PG} 
and additional Assumption \ref{assumption:augumented_state_visitation} that assumes positive visitation measure for any augmented state, suppose every agent $i\in\mathcal{N}$ follows the policy gradient dynamics \eqref{eq:PG}, which results in the update dynamics \eqref{eq:BayesianPG} for each each agent $i$, parent actions $a^{\mathcal{P}^i}$, and local action $a^i$, with $\eta\leq \frac{(1-\gamma)^3}{8N(r_{\rm max}-r_{\rm min})}$, then the converged fully-correlated BN policy 
$(\pi_{\theta^1_*}^1, \cdots, \pi_{\theta^N_*}^N,\mathcal{G})$ is an optimal policy.
\end{corollary}

\begin{figure}[t]
\begin{center}
\centerline{\includegraphics[width=\columnwidth]{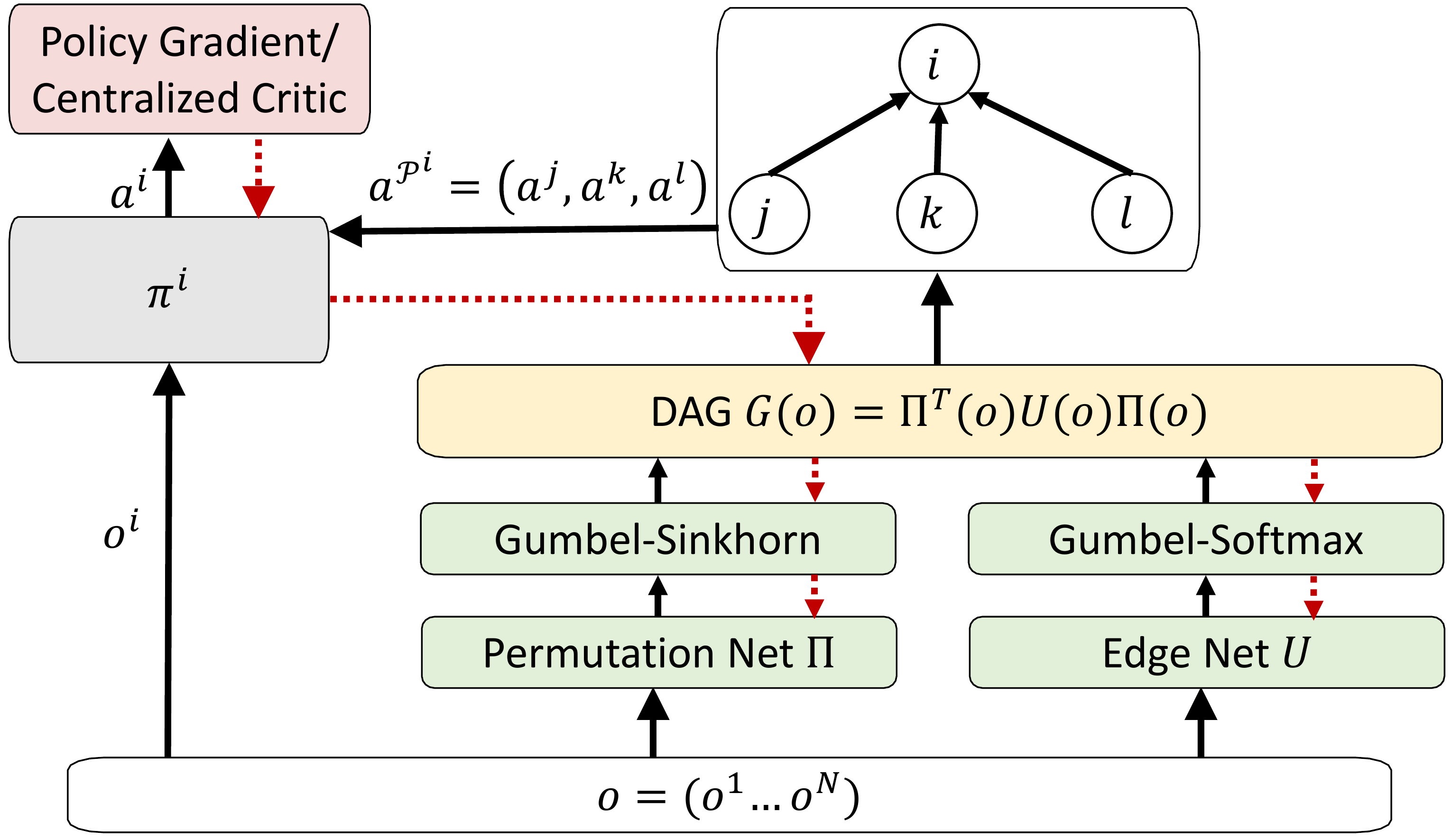}}
\caption{Our architecture for BN joint policy which includes each agent $i$'s policy $\pi^i$ and a differentiable DAG learner. DAG $G(o)=\Pi^T(o)U(o)\Pi(o)$ is generated by sending the joint local observation $o$ to Permutation Net $\Pi$ and Edge Net $U$.
Based on $G(o)$, agent $i$ requests actions $a^{\mathcal{P}^i}=(a^j,a^k,a^l)$ from its parents $(j,k,l)$, which, together with local observation $o^i$, are taken as input into agent $i$'s local policy $\pi^i$ to output $a^i$. During training, the gradient (shown in the red dotted lines) flows from $\pi^i$ to $G$, then to $\Pi$ and $U$. }
\label{fig:architecture_topology}
\end{center}
\vskip -0.3in
\end{figure}

\section{Practical Algorithm}
\label{sec:Practical Algorithm}
The convergence guarantee in Theorem \ref{theorem:Asymptotic convergence to Nash with gradient ascent} relies on global observability and the availability of the oracle value function, which is hard to apply in more complicated scenarios. In this section, we relax those assumptions and propose an end-to-end training framework which can augment any multi-agent actor-critic methods with a differentiable Bayesian network determining action dependencies among agents' local policies. 
Figure \ref{fig:architecture_topology} presents an overview of the our proposed neural architecture, consisting of the differentiable Bayesian network and the actor-critic networks as its main components that we describe below. 

\subsection{Differentiable Bayesian Network}
\label{sec:Differentiable Bayesian network}
The graph model $G$ takes the joint partial observation $o=\{o^i\}_{i\in\mathcal{N}}$ as input, and outputs a DAG $G(o)$ represented by an adjacency matrix, i.e., $G(o)[j,i]=1$ if and only if $j\in\mathcal{P}^i$. By using the same decomposition of DAG into the multiplication of permutation matrix and upper triangular matrix  \cite{charpentier2022differentiable} described in section \ref{sec:Releated Work}, $G$ consists of two sub-modules Permutation Net $\Pi$ and Edge Net $U$ which both takes the joint partial observation $o$ as input and output the logits $l_\Pi$ and $l_U$ for the permutation matrix and upper triangular matrix, respectively. We use the reparameterization trick Straight-Through
Gumbel-Softmax  \cite{jang2016categorical} and Gumbel-Sinkhorn  \cite{mena2018learning} to differentiably transform $l_\Pi$ and $l_U$ into the corresponding permutation matrix $\Pi(o)$ and upper triangular matrix $U(o)$. The resulting DAG $G(o)=\Pi^T(o)U(o)\Pi(o)$, where $\Pi(o)$ determines the topological ordering of the agents and $U(o)$ determines the structure of the outputted DAG.

\subsection{Actor and Critic Networks}
Our communication network is compatible with any multi-agent actor-critic architecture.
Our experiments mainly explore discrete actors which sample actions conditioning on local observation $o^i$ and parent actions $a^{\mathcal{P}^i}$, i.e., $a^i \sim \pi^i(\cdot|o^i,a^{\mathcal{P}^i})$. The critic takes the joint local observation or the environment provided global state as input. Both actor and critics are implemented by deep neural networks with details in the appendix \ref{table:Implementation details}.

\subsection{Training}
Critic $Q$ is trained to minimize TD error 
$\mathcal{L}_{\rm TD} =\mathbb{E}_{o_t,a_t,r_t,o_{t+1}}[(Q(o_t,a_t)-y_t)^2]
$,
where 
$o_t:=(o^1_t,...,o^N_t)$,
$a_t:=(a^1_t,...,a^N_t)$, 
and $y_t := r_t + \gamma Q(o_{t+1},a_{t+1})$ is the TD target. Actor $\pi^i$ can be updated by any multi-agent policy gradient algorithm, such as MAPPO $
    \mathcal{L}^i_{\rm actor} = \mathbb{E}_{o_t,a_t}[\text{log}\pi^i(a^i|o^i,a^{\mathcal{P}^i})A(o_t,a_t)]
$. Due to the differentiability enabled by Gumbel-Softmax and Gumbel-Sinkhorn, the gradient can flow from $\pi^i$ to DAG $G$, then to its sub-modules Permutation Net $\Pi$ and Edge Net $U$. The DAG Density of $G$ is defined as $\rho(G):=\frac{2}{|N(N-1)}\sum_{i,j\in\mathcal{N}} G[j,i]$, and is regularized by the term $\alpha|\rho(G)-\eta|$. This places a restriction on the sparsity of the learned DAG by rate $\eta$.

\section{Experiments}
\label{sec:Experiments}
Theorem \ref{theorem:Asymptotic convergence to Nash with gradient ascent} only guarantees that the policy gradient ascent converges to Nash, but does not guarantee the solution quality of the convergent. Our experiments, in the tabular softmax Bayesian setting, aim to see how well different (fixed) DAG topologies of the BN policy perform empirically and the reasons behind it. Then, in the sample-based setting, we want to see 1) How well our algorithm proposed in Section \ref{sec:Practical Algorithm} performed against baselines and ablations? 2) What are the potential meaning of the DAG learned by the context-aware differentiable DAG learner?
\subsection{Environments}
Our environments include
(1) Coordination Game, a small-size domain where we can afford computing exact policy gradient under tabular parameterization,
(2) Aloha, a domain where action correlations are intuitively helpful,
and (3) StarCraft II Micromanagement (SMAC), a common cooperative MARL benchmark that is more complicated.

\textbf{Coordination Game.} 
We use the version in \cite{chen2022convergence} with $N=2,3,5$ agents. The state space and action space are $\mathcal{S} = \mathcal{S}^1\times \cdots \times\mathcal{S}^N, \mathcal{A} = \mathcal{A}^1\times \cdots \times\mathcal{A}^N$, respectively, where $\forall i\leq N, \mathcal{S}^i \in\{0,1\},\mathcal{A}^i \in\{0,1\}$. It is a cooperative setting with the same reward for all the agents, which favors more agents in the same local state. The transition function for each agent $i$'s local state only depends on the local action: $P(s^i=0|a^i=0)=1-\epsilon$, $P(s^i=0|a^i=1)=\epsilon$, where $\epsilon=0.1$. The performance of the learned joint policy is measured by {\em price of anarchy} (POA) \cite{roughgarden2015intrinsic}, $\frac{V_\pi(\mu)}{\max_{\bar{\pi}} V_{\bar{\pi}}(\mu)}$, which is bounded in the range $[0,1]$. The convergence rate is captured by {\em Nash-Gap}, defined as $
    \mbox{Nash-gap}(\pi) := \max_{i}\left(\max_{\bar{\pi}^{i}} V_{\bar{\pi}^{i},\pi^{-i}}(\mu) - V_{\pi}(\mu)\right) 
$, where Nash policy has a Nash-Gap of zero.

\textbf{Aloha.} 
We use the version in  \cite{wang2022contextaware} with 10 agents (islands). 10 islands are stored in a $2\times5$ array, each of which has a backlog of messages to send. At each timestep, agents can either choose to send or not send. The goal is to send as many messages as possible without colliding with the ones sent by the neighboring islands. At each timestep, with a probability of 0.6, a new message can be generated for each agent. For each successfully sent message without collision, all agents receive a 0.1 reward, and a -10 reward if with collision. 

\textbf{StarCraft II Micromanagement (SMAC).} 
SMAC \cite{samvelyan2019starcraft} has become one of the most popular MARL benchmarks. We choose the \textit{Super Hard} scenarios 6h\_vs\_8z and MMM2 to evaluate our proposed algorithm, which has 6 agents and 10 agents, respectively. 
\subsection{Baselines}
As baselines to compare against our context-aware DAG topology learning to bring in correlations between local policies, we consider the following DAGs that are fixed during training (i.e., no context-awareness). 
The \textbf{Fully-correlated} baseline has DAG $(\mathcal{N},\{(j,i)|i>j\})$, which have the maximum number of $(N(N-1)/2)$ edges for any DAG. \textbf{Uncorrelated} has DAG $(\mathcal{N},\emptyset)$, i.e., product policy. \textbf{Line-correlated} has DAG $(\mathcal{N},\{(j,i)|i=j+1\})$.
The DAGs of all baselines have a topological ordering of $(1,2, \cdots,  N)$, i.e., $\Pi$ defined in Section \ref{sec:Differentiable Bayesian network} is fixed as the identity matrix. 

\begin{figure}[t]
\begin{center}
\centerline{\includegraphics[width=\columnwidth]{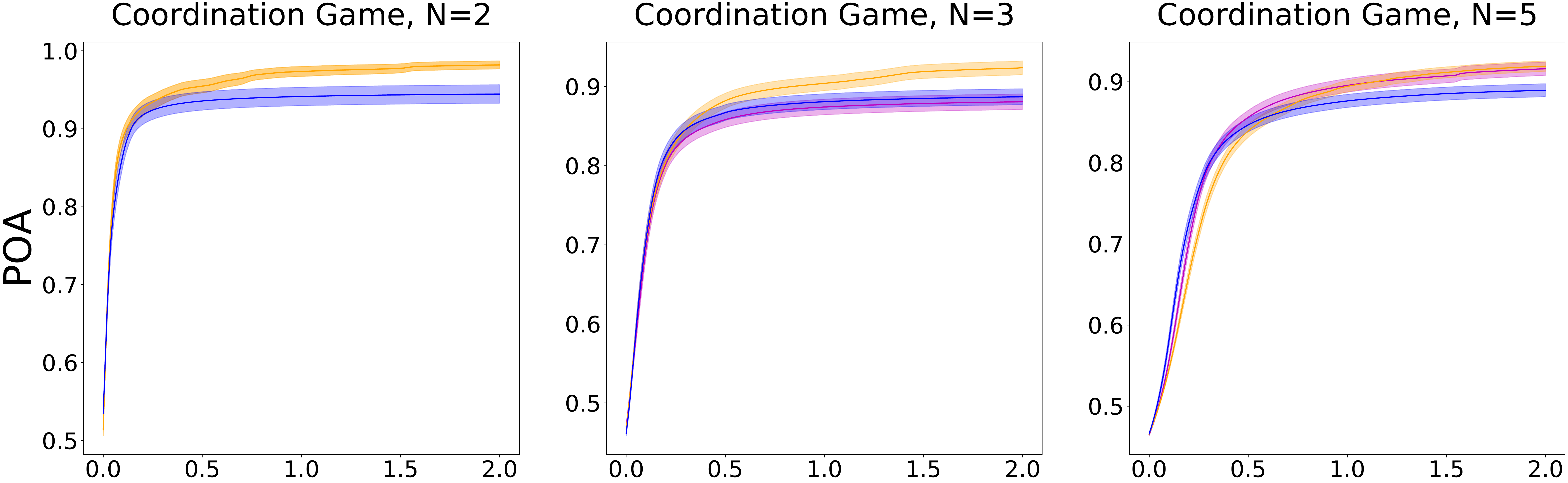}}
\centerline{\includegraphics[width=\columnwidth]{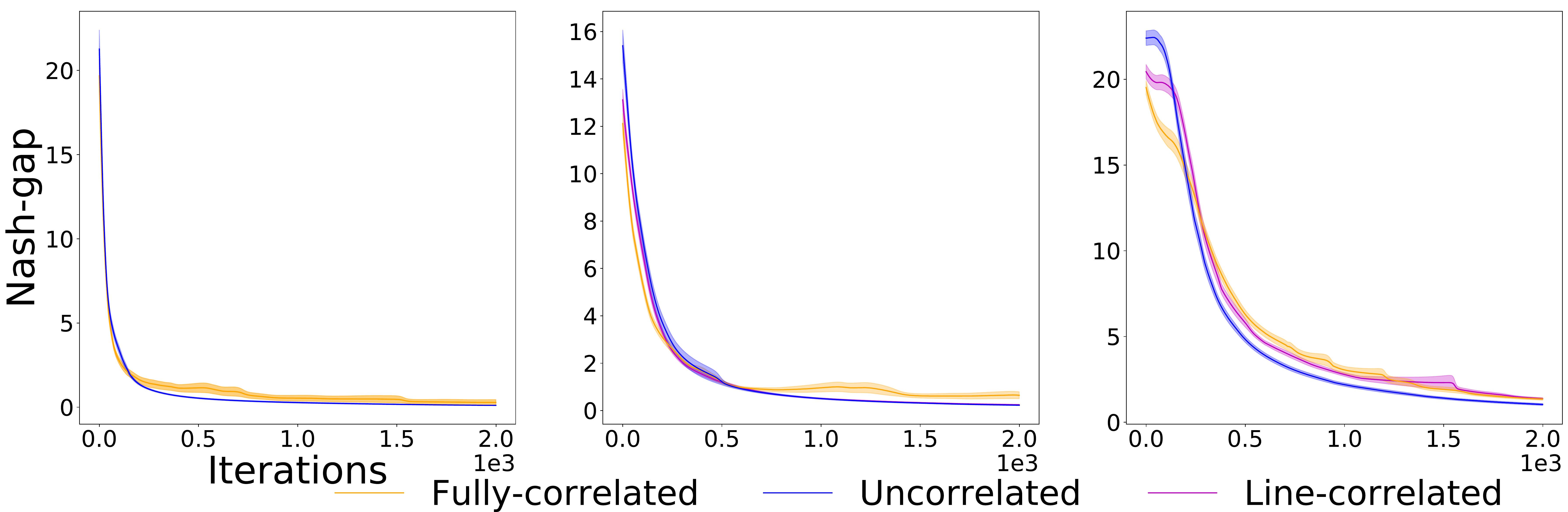}}
\caption{POA (top) and Nash-gap (bottom) under the tabular softmax BN policy gradient dynamics with various BN DAG topologies (means and standard errors over 50 random seeds). Note that with $N=2$ agents, Line-correlated and Fully-correlated are the same and thus have overlapping curves.}
\label{fig:Tabular-Byesian}
\end{center}
\vskip -0.3in
\end{figure}

\begin{figure*}[ht]
\begin{center}
\centerline{\includegraphics[width=\textwidth]{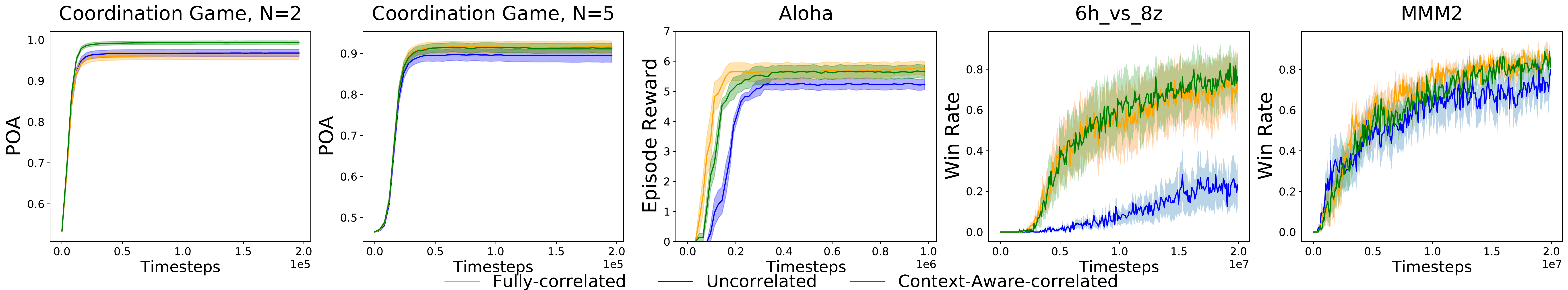}}
\vskip 0.1in
\centerline{\includegraphics[width=\textwidth]{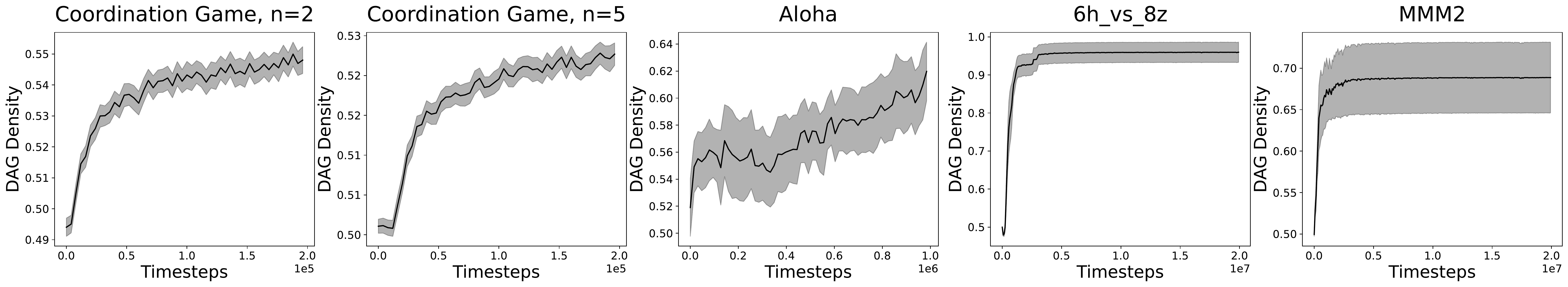}}
\caption{
\textit{Top:}
Performance of the learned context-aware-correlated against the Fully-correlated and Uncorrelated baselines. (means and standard errors over 50 random seeds for Coordination Game, 10 random seeds for Aloha, and 5 random seeds for 6h\_vs\_8z and MMM2.)
\textit{Bottom:}
Changes in DAG Density of the learned context-aware BN policy during training with no regularization.
}
\label{fig:Performance-DAG Density}
\end{center}
\vskip -0.3in
\end{figure*}

\subsection{Results of Fixed DAG Topologies with Tabular Exact Policy Gradients}
\label{sec:The effectiveness of different fixed topology in the tabular setting}
Figure \ref{fig:Tabular-Byesian} presents the POA and the Nash-gap of the algorithms under the tabular softmax parameterization with different DAG topologies. The results demonstrate that the Nash-gap indeed decreases and converges close to zero as proved in Theorem \ref{theorem:Asymptotic convergence to Nash with gradient ascent}. Fully-correlated consistently outperforms Line-correlated and Uncorrelated but does not converge to an optimal policy with POA of $1$, because Assumption \ref{assumption:augumented_state_visitation} is violated, i.e., some $(s, a^{\mathcal{P}^i})$ has visitation probability converges to zero. On the other hand, the convergence rate of Fully-correlated is the slowest and one possible reason is that it has the most number of parameters. Line-correlated has a similar performance to Fully-correlated in scenarios with $N=2,5$ agents, but it has poor performance in the scenario with 3 agents. This illustrates the fact that fixed DAG topology is not desirable in all scenarios and can degenerate to the performance of Uncorrelated.  


\subsection{Results of Context-Aware DAG Topology Learning with Multi-Agent Actor-Critic Methods}
In this section, we run experiments to compare our context-aware DAG topology against the baselines in Coordination Game, where we assume global observability, and in Aloha and SMAC, where we assume partial observability. We relax the requirement of only sharing local action to also include local observations when finding beneficial, based on the context-aware DAG. Specifically, based on the context-aware DAG, the experiments in Aloha share both local actions and observations, whereas the ones in SMAC only share local actions. We implement the algorithms based on MAPPO without recurrency, i.e., the model only incorporates information from the current timestep instead of from the whole trajectory, in both actor and critic. To have the uniform dimensionality required by the MLP-based actor, we handle the actions of the agents not selected by the DAG as the parents by padding dummy vectors of zeros.

\subsubsection{Coordination Game} We run the experiments in the Coordination Game with $N=2,5$ under full observability, and no regularization (i.e., $\alpha=0$), plotted in Figure \ref{fig:Performance-DAG Density}(top). Remarkably, the result in $N=2$ shows that our context-aware DAG learning outperforms the Fully-correlated. One possible explanation is that the dynamic graph leads to sufficient exploration of the augmented state defined in \ref{assumption:augumented_state_visitation}, and thus results in better performance. The context-aware DAG topology performs similarly to Fully-correlated in $N=5$, and both outperform Uncorrelated. As shown in Figure \ref{fig:Performance-DAG Density}(bottom), the density of the unregularized learned context-aware DAG is increasing in both $N=2$ and $N=5$ scenarios, from 50\% to 55\% and 50\% to 53\%, respectively. 

\subsubsection{Coordination Game} We run the experiments in the Coordination Game with $N=2,5$ under full observability, and no regularization (i.e., $\alpha=0$), plotted in Figure \ref{fig:Performance-DAG Density}(top). Remarkably, the result in $N=2$ shows that our context-aware DAG learning outperforms the Fully-correlated. One possible explanation is that the dynamic graph leads to sufficient exploration of the augmented state defined in \ref{assumption:augumented_state_visitation}, and thus results in better performance. The context-aware DAG topology performs similarly to Fully-correlated in $N=5$, and both outperform Uncorrelated. As shown in Figure \ref{fig:Performance-DAG Density}(bottom), the density of the unregularized learned context-aware DAG is increasing in both $N=2$ and $N=5$ scenarios, from 50\% to 55\% and 50\% to 53\%, respectively.

\subsubsection{Aloha}
We run the experiments in the Aloha with $N=10$ under partial observability (each agent observes the backlog of its own messages), and no sparsity regularization (i.e., $\alpha=0$).
The results in Figure \ref{fig:Performance-DAG Density}(top) show that our context-aware DAG learning performs comparably to Fully-correlated, and both outperform Uncorrelated. Note that the initial policy at timestep $0$ with a random initialization will generate collisions resulting in large negative rewards. The policy will soon learn to avoid collisions, and we only show the performance when the policy can generate positive rewards. As shown in Figure \ref{fig:Performance-DAG Density}(bottom), the density of the unregularized learned context-aware DAG is also increasing from 50\% to 62\% which is larger than the ones learned in the Coordination Game. This suggests that the action dependencies in Aloha may be more important than the ones in the Coordination Game.

\textbf{Analysis: Learned DAG topologies.}
\begin{figure}[b]
\begin{center}
\centerline{\includegraphics[width=.97\columnwidth]{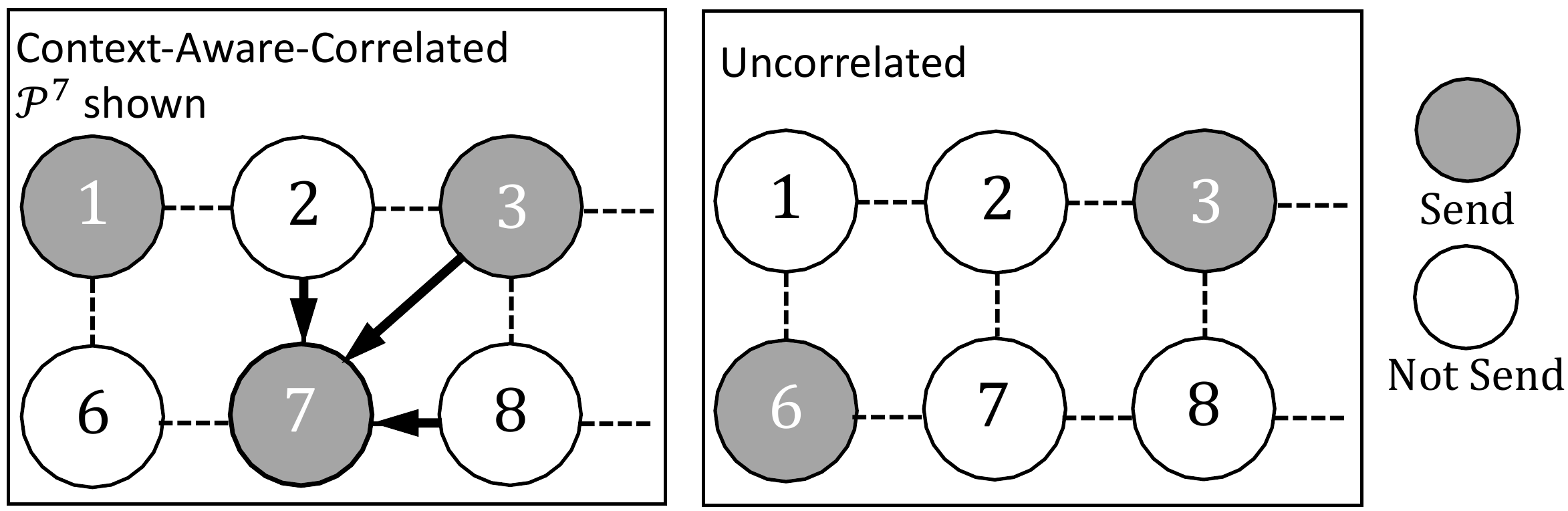}}
\caption{Learned DAG topology in Aloha. Only $\mathcal{P}^7$ is shown.}
\label{Aloha-Visualization}
\end{center}
\vskip -0.3in
\end{figure}
In the timestep shown in Figure \ref{Aloha-Visualization}, each agent has a backlog of one message to send. For Context-Aware-Correlated, guided by the learned topology, agent 7 obtains the extra parent action (and observation) dependencies from nearby agents 2 and agent 8 , which do not send, and far-away agent 3 which sends but causes no collision. Thus, agent 7 is therefore more confident to send its message. For Uncorrelated, agents need to be more careful to avoid collisions. Both agents 2 and 7 choose not to send in this case to make it safer for agent 3 and 6 to send. This results in one less message sent for the shown agents.

\subsubsection{SMAC}
We run the experiments in the \textit{Super Hard} SMAC scenario 6h\_vs\_8z and MMM2 under partial observability with no last actions stored, plotted in Figure \ref{fig:Performance-DAG Density}(top). 6h\_vs\_8z and MMM2 are noisier than the Coordination Game and Aloha, and we find no benefit of using permutation matrix $\Pi$ to change the topological ordering. Therefore, we use a fixed topological ordering where $\Pi$ is the identity matrix. The action dependencies in both scenarios are crucial, as we can see in Figure \ref{fig:Performance-DAG Density}(bottom) that the unregularized context-aware graph degenerates to an almost Full-Dependency graph in 6h\_vs\_8z and densely correlated graph with around 70\% DAG density in MMM2, respectively. Therefore, we regularize it to control the DAG density with an annealing strategy, which gradually decreases threshold $\eta$ and increases regularization weight $\alpha$. Specifically, in the first $a$\% training steps, the sparsity threshold $\eta$ is set to 1, which encourages the agents to learn that the action dependencies are useful. Then, from $a$\% total training steps to $b$\% total training steps, we decrease sparsity threshold $\eta$ from 1 to 0 uniformly in $l_\eta$ times. From $b$\% total training steps to $c$\% total training steps, we uniformly increase in $l_\alpha$ times the regularization weight $\alpha$ from 0.1 to 1 in 6h\_vs\_8z and 0.05 to 0.5 in MMM2. 

As shown in Figure \ref{Annealing}, from 0\% to a\% total training steps, the performance is similar to the Fully-correlated baseline in both scenarios, with DAG density quickly becoming close to 1. From a\% to b\%, as we decrease sparsity threshold $\eta$, the performance fluctuates but still be much better than Uncorrelated. From b\% to c\%, we increase the regularization weight $\alpha$. For 6h\_vs\_8z, the performance decreases quickly close to Uncorrelated, but then recovers quickly to be better than Uncorrelated. For MMM2, the transition is more smooth and the performance consistently beat the Uncorrelated baseline. This multi-phase regularization strategy results in purely uncorrelated policies as shown in Figure \ref{fig:Performance-DAG Density}(bottom), but achieves better performance than Uncorrelated, which is trained with purely uncorrelated policies during the whole training phase as shown in Figure \ref{Annealing}.

\begin{figure}[t]
\begin{center}
\centerline{\includegraphics[width=\columnwidth]{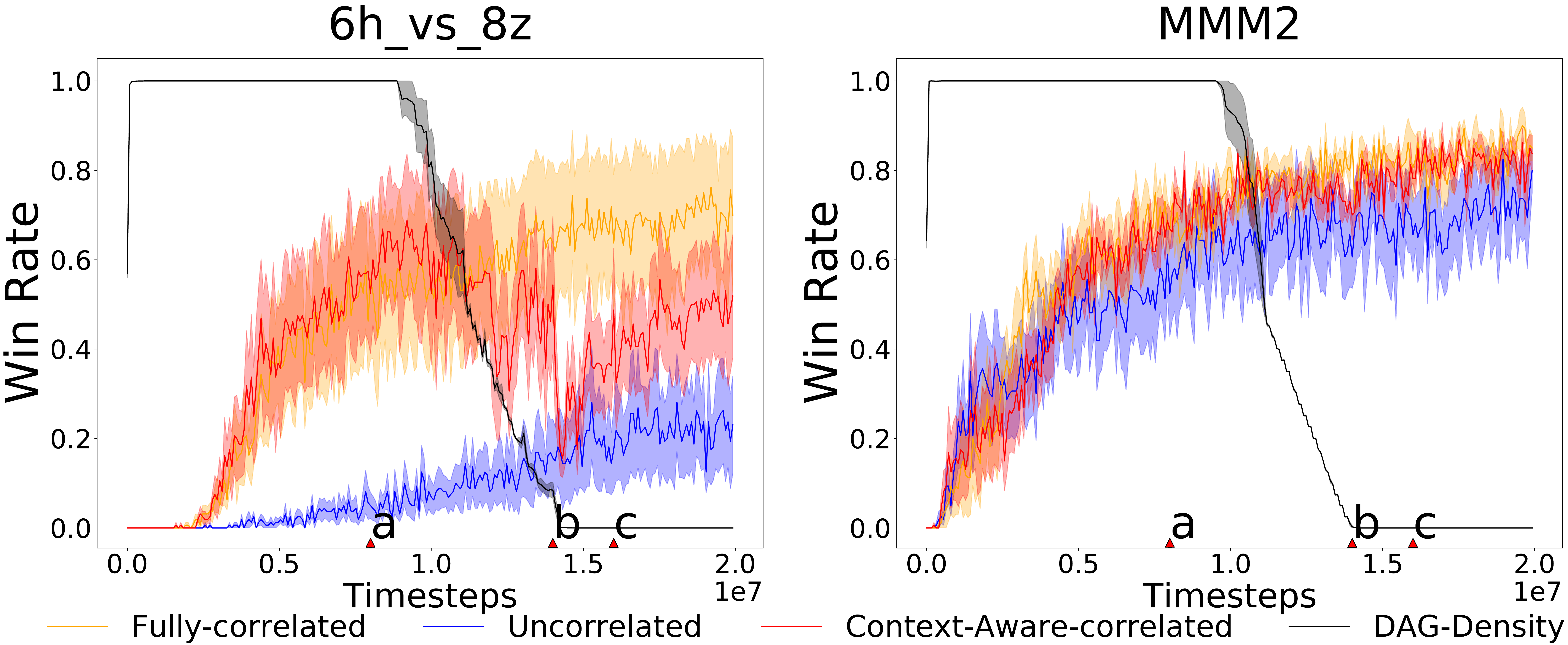}}
\caption{The performance of the sparsity regularized context-aware-correlated (with density annealing) in 6h\_vs\_8z and MMM2, against the Fully-correlated and Uncorrelated baselines. The black lines show the Changes in DAG Density of sparsity regularized context-aware-correlated (with density annealing) during training. }
\label{Annealing}
\end{center}
\vskip -0.3in
\end{figure}

\textbf{Analysis: Visibility.}
The dependency on the allies which are not visible is meaningless. Since action dependencies in both 6h\_vs\_8z and MMM2 are important, one simple strategy that an agent can learn to maintain a good performance while decreasing the DAG density is to only output dependency on the visible agents. This is indeed the case, with the percentage of the visible agents that an agent wants to depend on almost consistently increasing in Figure \ref{visiblity}. 
\begin{figure}[hb]
\begin{center}
\centerline{\includegraphics[width=\columnwidth]{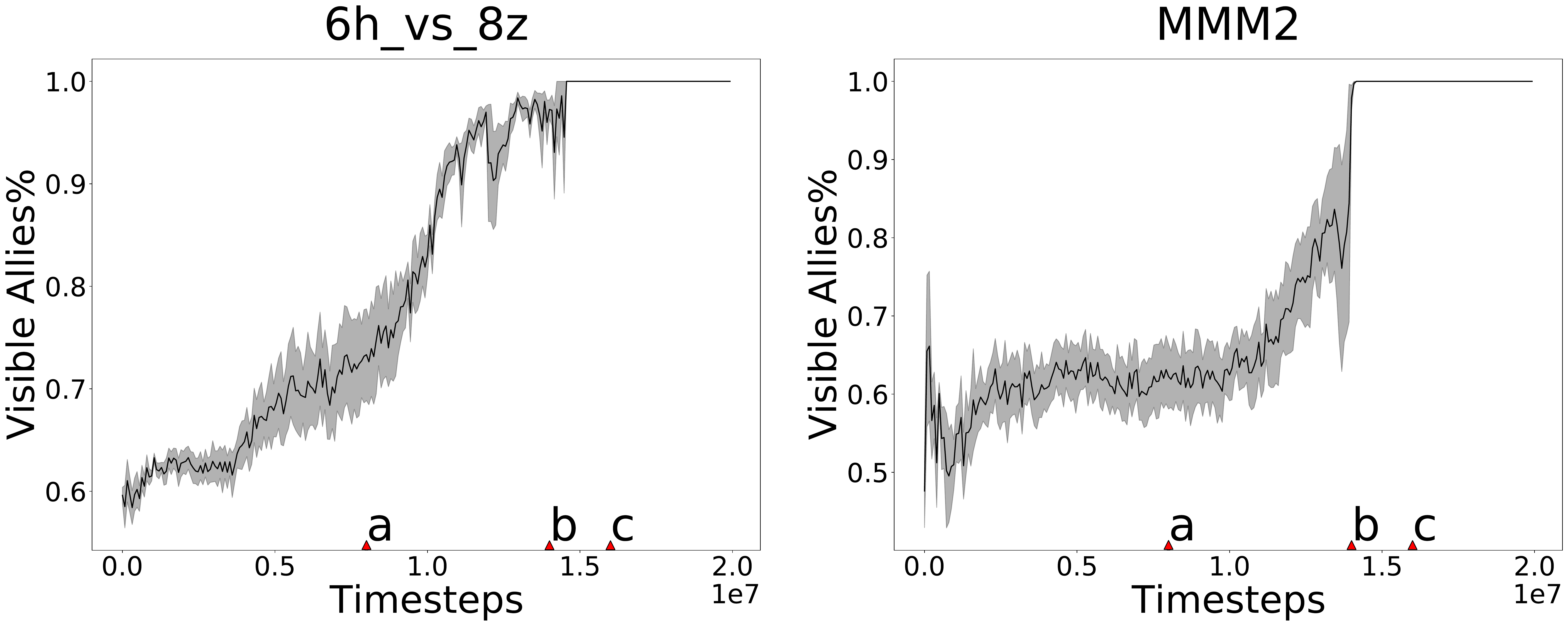}}
\caption{Visibility of allies during training
(context-aware-correlated with density annealing) in 6h\_vs\_8z and MMM2.}
\label{visiblity}
\end{center}
\vskip -0.3in
\end{figure}

\textbf{Analysis: Average health.}
As shown in Figure \ref{health}, for 6h\_vs\_8z, agents 5 and 6 tend to depend on the actions of agents with a relatively low health bar, while agent 2 tends to depend on the actions of agents with relatively high health. This may be due to that we fixed the topological ordering of $(1, \cdots, 6)$, so agents 5 and 6 can potentially have more dependencies, and it learns to depend on the actions of agents with low health. On the other hand, agent $2$ can only depend on the action of agent $1$ which may not always have low health. For MMM2, health does not differentiate agents' selections of parent actions until in the middle of a\% to b\%, where the increase of the regularization causes all agents to depend on actions of agents with relatively high health, with agent 7 to the extreme. 

\textbf{Analysis: Average distance.}
As shown in Figure \ref{distance}, in both scenarios, agent 6 tends to depend on the actions of agents in relatively long distances. In 6h\_vs\_8z, agent 2 tends to depend on the actions of agents with relatively short distances, while distance is relatively irrevelant for agents except agent 6 for the selection of parent actions. This also may be due to that we fixed the topological ordering of agent $(1, \cdots, N)$. Agent 6 can potentially have more dependencies, so it learns to depend on the actions of agents that are far away. On the other hand, agent $2$ can only depend on the action of agent $1$ which may be nearby sometimes. For MMM2, agent 8 consistently depends on the actions of agents from relatively longer distances, whereas agent 5 behaves the opposite.

\begin{figure}[t]
\begin{center}
\centerline{\includegraphics[width=\columnwidth]{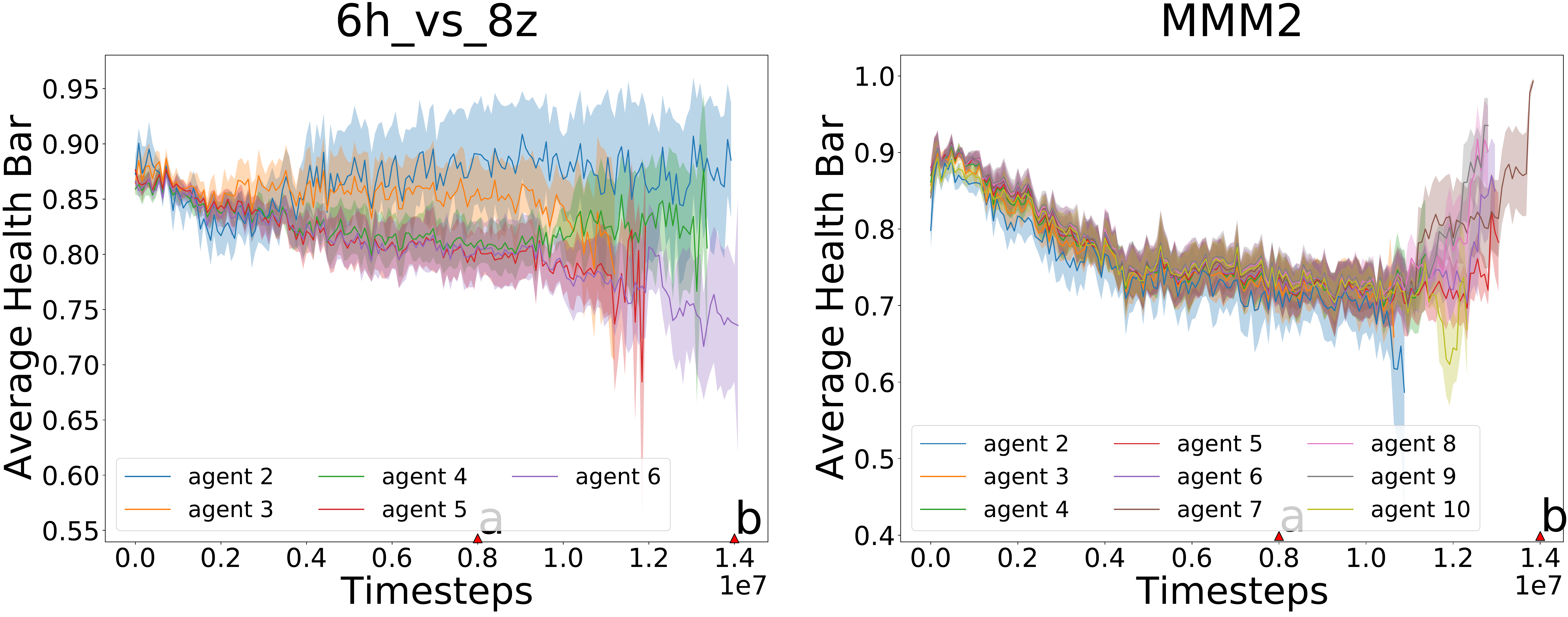}}
\caption{
Average health during training
(context-aware-correlated with sparsity annealing) in 6h\_vs\_8z and MMM2.
}
\label{health}
\end{center}
\vskip -0.3in
\end{figure}

\begin{figure}[ht]
\begin{center}
\centerline{\includegraphics[width=\columnwidth]{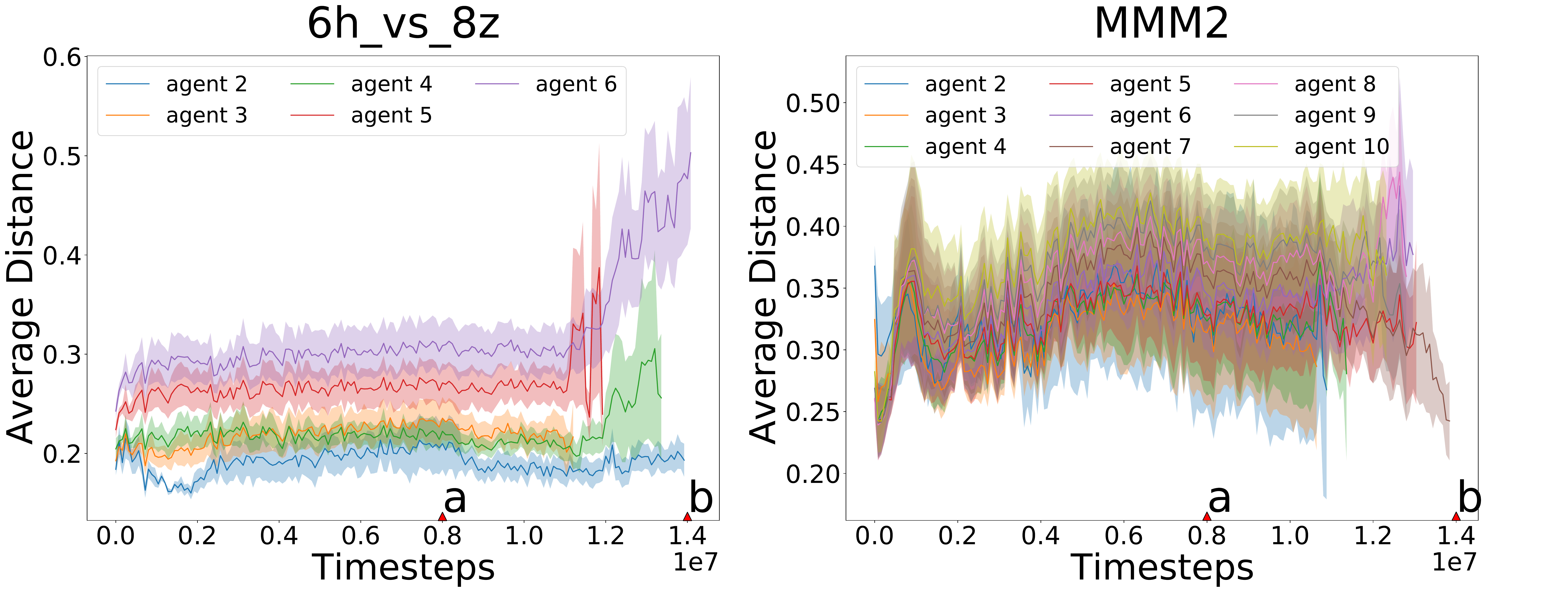}}
\caption{
Average parent distance during training
(context-aware-correlated with sparsity annealing) in 6h\_vs\_8z and MMM2.
}
\label{distance}
\end{center}
\vskip -0.3in
\end{figure}

\textbf{Analysis: Emergence of multi-modality for BN policy} Previous works \cite{baker2019emergent,lowe2019pitfalls,tang2021discovering} show that the emergence of diverse behaviors is prevalent in many multi-agent reinforcement learning problems. A recent work \cite{fu2022revisiting} shows the benefit of learning a multi-modal policy. Here we analyze the emergence of multi-modality for BN policy learned with the annealing strategy. To quantify multi-modality, we measure the KL Divergence between the distribution of the BN policy and the distribution of the same BN policy with empty DAG. The result in Figure \ref{multi-modality} shows that in 6h\_vs\_8z,
agent 6 with most possible parent actions has the largest multi-modality, whereas agent 2 with least possible parent actions except agent 1 has the smallest multi-modality. In MMM2, agent 8 emerges with the largest multi-modality, whereas agent 2 with least possible parent actions except agent 1 has the smallest multi-modality. It is also shown in both scenarios that increasing DAG density regularization also decreases multi-modality, where the purely decentralized one has zero multi-modality.

\begin{figure}[H]
\begin{center}
\centerline{\includegraphics[width=\columnwidth]{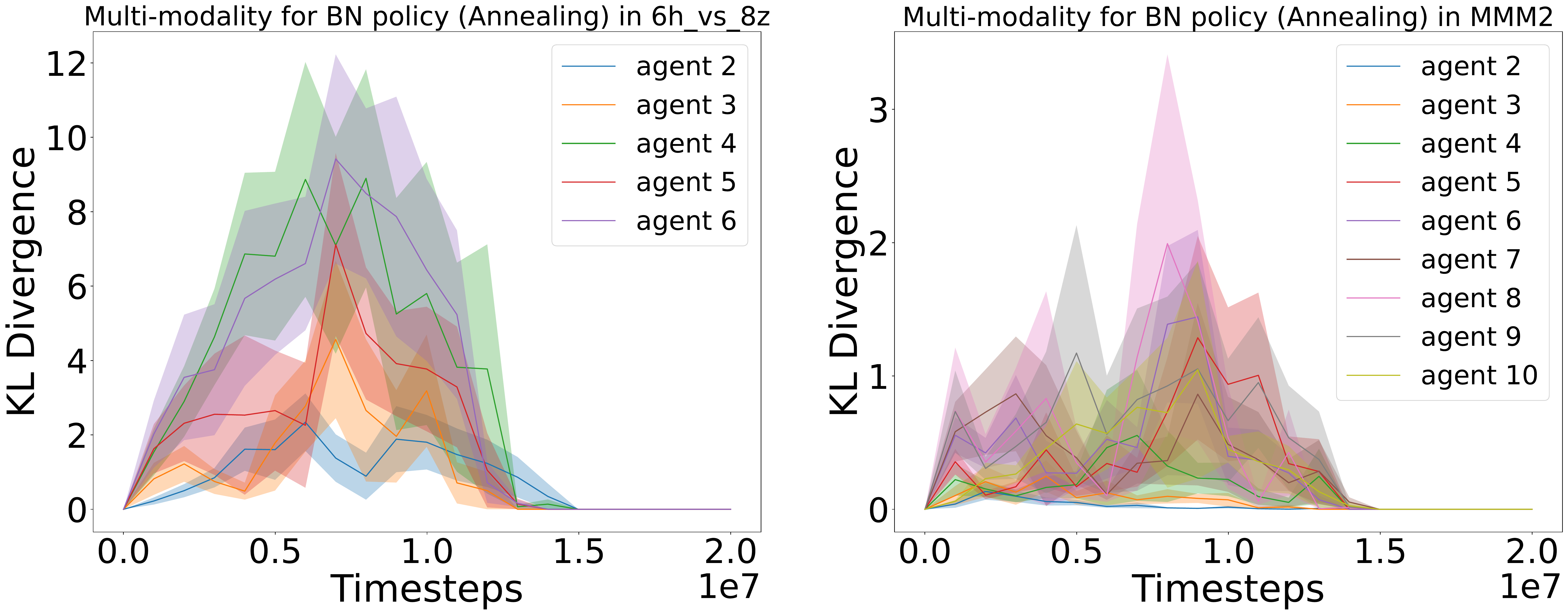}}
\caption{
Emergence of multi-modality for BN policy (annealing) in SMAC.
}
\label{multi-modality}
\end{center}
\vskip -0.3in
\end{figure}

\section{Conclusion}
\label{sec:Conclusion}
In this paper, we have motivated action correlations for cooperative MARL and proposed the notion of BN joint policy to introduce correlations.
We have then derived the BN policy gradient formula and proved the convergence to Nash policy asymptotically under the tabular softmax BN policy parameterization. Further, we have proposed a practical algorithm to adapt any multi-agent actor-critic method to realize the BN joint policy and empirically demonstrated the benefits of the proposed method. 

\section*{Acknowledgments}
Dingyang Chen acknowledges funding support from NSF award IIS-2154904.
Qi Zhang acknowledges funding support from NSF award IIS-2154904 and NSF CAREER award 2237963. Any opinions, findings, conclusions, or recommendations expressed here are those of the authors and do not necessarily reflect the views of the sponsors.

\bibliography{example_paper}
\bibliographystyle{icml2023}

\newpage
\appendix
\onecolumn
\section{Proof of Theorem \ref{theorem:Asymptotic convergence to Nash with gradient ascent} and Corollary \ref{theorem:Asymptotic convergence to optimal with gradient ascent}}
\label{Proof of Theorem}

\begin{lemma}
\label{helper lemma for smoothness of V}
\begin{equation}\label{eqn:bound on Q}|Q^{\pi_{\Tilde{\theta}}}(s,a)-Q^{\pi_{{\theta}}}(s,a)|\leq \frac{r_{\rm max}-r_{\rm min}}{(1-\gamma)^2}\max_s\|\pi_{\Tilde{\theta_s}}-\pi_{{\theta_s}}\|_1\end{equation}
\begin{equation}\label{eqn:bound on EQ}|\overline{A}^{\pi_{\Tilde{\theta}},i}(s,a)-\overline{A}^{\pi_{{\theta}},i}(s,a)|\leq \frac{2(r_{\rm max}-r_{\rm min})}{(1-\gamma)^2}\max_s\|\pi_{\Tilde{\theta_s}}-\pi_{{\theta_s}}\|_1\end{equation}
where $\overline{A}^{\pi_{\theta,i}}(s,a)=Q^{\pi_{{\theta}}}(s,a)-\E_{\bar{a}^{i}\sim\pi_{{\theta^i}}^i(\cdot|s,a^{\mathcal{P}^i})}Q^{\pi_{{\theta}}}(s,\bar{a}^{i},a^{-i})$.
\end{lemma}

\begin{proof}
Equation (\ref{eqn:bound on Q}) is proved in lemma 32 in \cite{zhang2022effect}.
\\For Equation (\ref{eqn:bound on EQ}), 
$$|A^{\pi_{\Tilde{\theta}},i}(s,a)-A^{\pi_{{\theta}},i}(s,a)|=|Q^{\pi_{\Tilde{\theta}}}(s,a)-\E_{\bar{a}^{i}\sim\pi_{\Tilde{\theta}^i}^i(\cdot|s,a^{\mathcal{P}^i})}Q^{\pi_{\Tilde{\theta}}}(s,\bar{a}^{i},a^{-i})-(Q^{\pi_{{\theta}}}(s,a)-\E_{\bar{a}^{i}\sim\pi_{{\theta^i}}^i(\cdot|s,a^{\mathcal{P}^i})}Q^{\pi_{{\theta}}}(s,\bar{a}^{i},a^{-i}))|$$
$$\leq |Q^{\pi_{\Tilde{\theta}}}(s,a)-Q^{\pi_{{\theta}}}(s,a)|+|\E_{\bar{a}^{i}\sim\pi_{\Tilde{\theta}^i}^i(\cdot|s,a^{\mathcal{P}^i})}Q^{\pi_{\Tilde{\theta}}}(s,\bar{a}^{i},a^{-i})-\E_{\bar{a}^{i}\sim\pi_{{\theta^i}}^i(\cdot|s,a^{\mathcal{P}^i})}Q^{\pi_{{\theta}}}(s,\bar{a}^{i},a^{-i})|$$
$$\leq |Q^{\pi_{\Tilde{\theta}}}(s,a)-Q^{\pi_{{\theta}}}(s,a)|+\max_a|Q^{\pi_{\Tilde{\theta}}}(s,a)-Q^{\pi_{{\theta}}}(s,a)|$$
By Equation (\ref{eqn:bound on Q}),
$$\leq \frac{r_{\rm max}-r_{\rm min}}{(1-\gamma)^2}\max_s\|\pi_{\Tilde{\theta_s}}-\pi_{{\theta_s}}\|_1+\frac{r_{\rm max}-r_{\rm min}}{(1-\gamma)^2}\max_s\|\pi_{\Tilde{\theta_s}}-\pi_{{\theta_s}}\|_1=\frac{2(r_{\rm max}-r_{\rm min})}{(1-\gamma)^2}\max_s\|\pi_{\Tilde{\theta_s}}-\pi_{{\theta_s}}\|_1$$
\end{proof}

\begin{lemma}
\label{Bound of BN PG}
$\|\nabla_{\theta^i}V(\Tilde{\theta})-\nabla_{\theta^i}V(\theta)\|_1 \leq\frac{8(r_{\rm max}-r_{\rm min})}{(1-\gamma)^3}\sum_{i=1}^N \|{\Tilde{\theta}^i}-{{\theta}^i}\|_2$
\end{lemma}
\begin{proof}
$$\|\nabla_{\theta^i}V(\Tilde{\theta})-\nabla_{\theta^i}V(\theta)\|_1$$
$$=\frac{1}{1-\gamma}\sum_{s,a^{\mathcal{P}^i},a^i}|d_\mu^{\pi_{\Tilde{\theta}}}(s,a^{\mathcal{P}^i})\pi_{{\Tilde{\theta}}^{i}}^i(a^i|s,a^{\mathcal{P}^i})A^{\pi_{\Tilde{\theta}},i}(s,a^{\mathcal{P}^i},a^i)-d_\mu^{\pi_\theta}(s,a^{\mathcal{P}^i})\pi_{\theta^{i}}^i(a^i|s,a^{\mathcal{P}^i})A^{\pi_\theta,i}(s,a^{\mathcal{P}^i},a^i)|$$
$$=\frac{1}{1-\gamma}\sum_{s,a^{\mathcal{P}^i},a^i}|d_\mu^{\pi_{\Tilde{\theta}}}(s,a^{\mathcal{P}^i})\pi_{{\Tilde{\theta}}^{i}}^i(a^i|s,a^{\mathcal{P}^i})*$$
$$\sum_{a^{-\mathcal{P}^i_+}}\pi_{{\Tilde{\theta}}}(a^{-\mathcal{P}^i_+}|s,a^{\mathcal{P}^i},a^i)\big(Q^{\pi_{\Tilde{\theta}}}(s,a^{\mathcal{P}^i},a^i,a^{-\mathcal{P}^i_+})-\E_{\bar{a}^{i}\sim\pi_{\Tilde{\theta}^i}^i(\cdot|s,a^{\mathcal{P}^i})}Q^{\pi_{\Tilde{\theta}}}(s,a^{\mathcal{P}^i},\bar{a}^{i},a^{-\mathcal{P}^i_+})\big)-$$
$$d_\mu^{\pi_{{\theta}}}(s,a^{\mathcal{P}^i})\pi_{{{\theta}}^{i}}^i(a^i|s,a^{\mathcal{P}^i})\sum_{a^{-\mathcal{P}^i_+}}\pi_{{{\theta}}}(a^{-\mathcal{P}^i_+}|s,a^{\mathcal{P}^i},a^i)\big(Q^{\pi_{{\theta}}}(s,a^{\mathcal{P}^i},a^i,a^{-\mathcal{P}^i_+})-\E_{\bar{a}^{i}\sim\pi_{{\theta^i}}^i(\cdot|s,a^{\mathcal{P}^i})}Q^{\pi_{{\theta}}}(s,a^{\mathcal{P}^i},\bar{a}^{i},a^{-\mathcal{P}^i_+})\big)|$$
$$\leq\frac{1}{1-\gamma}\sum_{s,a}|d_\mu^{\pi_{\Tilde{\theta}}}(s)\pi_{\Tilde{\theta}}(a|s)\big(Q^{\pi_{\Tilde{\theta}}}(s,a)-\E_{\bar{a}^{i}\sim\pi_{\Tilde{\theta}^i}^i(\cdot|s,a^{\mathcal{P}^i})}Q^{\pi_{\Tilde{\theta}}}(s,\bar{a}^{i},a^{-i})\big)-$$
$$d_\mu^{\pi_{{\theta}}}(s)\pi_{{\theta}}(a|s)\big(Q^{\pi_{{\theta}}}(s,a)-\E_{\bar{a}^{i}\sim\pi_{{\theta^i}}^i(\cdot|s,a^{\mathcal{P}^i})}Q^{\pi_{{\theta}}}(s,\bar{a}^{i},a^{-i})\big)|$$
\\Denote $\overline{A}^{\pi_{\theta,i}}(s,a)=Q^{\pi_{{\theta}}}(s,a)-\E_{\bar{a}^{i}\sim\pi_{{\theta^i}}^i(\cdot|s,a^{\mathcal{P}^i})}Q^{\pi_{{\theta}}}(s,\bar{a}^{i},a^{-i})$,
$$=\frac{1}{1-\gamma}\sum_{s,a}|d_\mu^{\pi_{\Tilde{\theta}}}(s)\pi_{\Tilde{\theta}}(a|s)\overline{A}^{\pi_{\Tilde{\theta},i}}(s,a)-d_\mu^{\pi_{{\theta}}}(s)\pi_{{\theta}}(a|s)\overline{A}^{\pi_{{\theta}},i}(s,a)|$$
$$\leq\frac{1}{1-\gamma}\bigg(\sum_{s,a}|d_\mu^{\pi_{\Tilde{\theta}}}(s)\pi_{\Tilde{\theta}}(a|s)-d_\mu^{\pi_{\Tilde{\theta}}}(s)\pi_{\Tilde{\theta}}(a|s)|\big|\overline{A}^{\pi_{\Tilde{\theta},i}}(s,a)\big|+\sum_{s,a}d_\mu^{\pi_{{\theta}}}(s)\pi_{{\theta}}(a|s)\big|\overline{A}^{\pi_{\Tilde{\theta},i}}(s,a)-\overline{A}^{\pi_{{\theta}},i}(s,a)\big|\bigg)$$
Since $\overline{A}^{\pi_{\theta,i}}(s,a)\leq \frac{2(r_{\rm max}-r_{\rm min})}{1-\gamma}$,
$$\leq\frac{1}{1-\gamma}\bigg(\sum_{s,a}\frac{2(r_{\rm max}-r_{\rm min})}{1-\gamma}|d_\mu^{\pi_{\Tilde{\theta}}}(s)\pi_{\Tilde{\theta}}(a|s)-d_\mu^{\pi_{\Tilde{\theta}}}(s)\pi_{\Tilde{\theta}}(a|s)+\max_{s,a}\big|\overline{A}^{\pi_{\Tilde{\theta},i}}(s,a)-\overline{A}^{\pi_{{\theta}},i}(s,a)\big|\bigg)$$
By Equation (\ref{eqn:bound on EQ}), 
$$\leq\frac{1}{1-\gamma}\bigg(\sum_{s,a}\frac{2(r_{\rm max}-r_{\rm min})}{1-\gamma}|d_\mu^{\pi_{\Tilde{\theta}}}(s)\pi_{\Tilde{\theta}}(a|s)-d_\mu^{\pi_{\Tilde{\theta}}}(s)\pi_{\Tilde{\theta}}(a|s)+\frac{2(r_{\rm max}-r_{\rm min})}{(1-\gamma)^2}\max_s\|\pi_{\Tilde{\theta_s}}-\pi_{{\theta_s}}\|_1\bigg)$$
By corollary (35) in \cite{zhang2022effect}, 
$$\leq\frac{1}{1-\gamma}\bigg(\frac{2(r_{\rm max}-r_{\rm min})}{(1-\gamma)^2}\max_s\|\pi_{\Tilde{\theta_s}}-\pi_{{\theta_s}}\|_1+\frac{2(r_{\rm max}-r_{\rm min})}{(1-\gamma)^2}\max_s\|\pi_{\Tilde{\theta_s}}-\pi_{{\theta_s}}\|_1\bigg)$$
$$=\frac{4(r_{\rm max}-r_{\rm min})}{(1-\gamma)^3}\max_s\|\pi_{\Tilde{\theta_s}}-\pi_{{\theta_s}}\|_1$$

$$\leq\frac{4(r_{\rm max}-r_{\rm min})}{(1-\gamma)^3}\max_s\sum_{i,a^{\mathcal{P}^i}}\|\pi_{\Tilde{\theta}^i_{s,a^{\mathcal{P}^i}}}-\pi_{{\theta}^i_{s,a^{\mathcal{P}^i}}}\|_1$$
By 
corollary (37) in \cite{zhang2022effect}, 
$$\leq\frac{8(r_{\rm max}-r_{\rm min})}{(1-\gamma)^3}\max_s\sum_{i,a^{\mathcal{P}^i}}\|{\Tilde{\theta}^i_{s,a^{\mathcal{P}^i}}}-{{\theta}^i_{s,a^{\mathcal{P}^i}}}\|_2$$
$$\leq\frac{8(r_{\rm max}-r_{\rm min})}{(1-\gamma)^3}\max_s\sum_{i=1}^N\|{\Tilde{\theta}^i_s}-{{\theta}^i_s}\|_2$$
$$\leq\frac{8(r_{\rm max}-r_{\rm min})}{(1-\gamma)^3}\sum_{i=1}^N \|{\Tilde{\theta}^i}-{{\theta}^i}\|_2$$

\end{proof}

\begin{lemma}[Smoothness of $V$ under tabular Baysian softmax]
\label{lemma:smoothness_of_V}
$$\|\nabla_{\theta}V(\Tilde{\theta})-\nabla_{\theta}V(\theta)\|_2\leq \frac{8N(r_{\rm max}-r_{\rm min})}{(1-\gamma)^3}\|\Tilde{\theta}-{\theta}\|_2$$
\end{lemma}
\begin{proof}
$$\|\nabla_{\theta}V(\Tilde{\theta})-\nabla_{\theta}V(\theta)\|_2^2=\sum_{i=1}^N\|\nabla_{\theta^i}V(\Tilde{\theta})-\nabla_{\theta^i}V(\theta)\|_2^2$$
$$\leq\sum_{i=1}^N\|\nabla_{\theta^i}V(\Tilde{\theta})-\nabla_{\theta^i}V(\theta)\|_1^2$$
By lemma (\ref{Bound of BN PG}),
$$\leq\sum_{i=1}^N\bigg(\frac{8(r_{\rm max}-r_{\rm min})}{(1-\gamma)^3}\sum_{j=1}^N \|{\Tilde{\theta}^j}-{{\theta}^j}\|_2\bigg)^2$$
$$=\frac{64N(r_{\rm max}-r_{\rm min})^2}{(1-\gamma)^6}\bigg(\sum_{i=1}^N \|{\Tilde{\theta}^i}-{{\theta}^i}\|_2\bigg)^2$$
$$\leq\frac{64N^2(r_{\rm max}-r_{\rm min})^2}{(1-\gamma)^6}\sum_{i=1}^N \|{\Tilde{\theta}^i}-{{\theta}^i}\|_2^2$$
$$=\frac{64N^2(r_{\rm max}-r_{\rm min})^2}{(1-\gamma)^6}\|{\Tilde{\theta}}-{{\theta}}\|_2^2$$
Therefore,
$$\|\nabla_{\theta}V(\Tilde{\theta})-\nabla_{\theta}V(\theta)\|_2\leq \frac{8N(r_{\rm max}-r_{\rm min})}{(1-\gamma)^3}\|\Tilde{\theta}-{\theta}\|_2$$
\end{proof}
\begin{lemma}
\label{BN_PG}
For a Baysian policy defined by $\mathcal{G}$, $\forall s, a^{\mathcal{P}^i}, a^i$,  
$$\frac{\partial V^{\pi_\theta}(\mu)}{\partial \theta_{s,a^{\mathcal{P}^i}, a^i}^{i}}
=\frac{1}{1-\gamma}d_\mu^{\pi_\theta}(s,a^{\mathcal{P}^i})\pi_{\theta^{i}}^i(a^i|s,a^{\mathcal{P}^i})A^{\pi_\theta,i}(s,a^{\mathcal{P}^i},a^i)$$
,where $d_\mu^{\pi_\theta}(s,a^{\mathcal{P}^i})=d_\mu^{\pi_\theta}(s)\textstyle\sum_{a^{-\mathcal{P}^i}}\pi_{\theta}(a^{-\mathcal{P}^i},a^{\mathcal{P}^i}|s)$,
$A^{\pi_\theta,i}(s,a^{\mathcal{P}^i},a^i) = Q^{\pi_\theta,i}(s,a^{\mathcal{P}^i},a^i)-Q^{\pi_\theta,i}(s,a^{\mathcal{P}^i})$, $Q^{\pi_\theta,i}(s,a^{\mathcal{P}^i},a^i)=\E_{\bar{a}^{-\mathcal{P}^i_+}\sim\pi_{\theta}(\cdot|s,a^{\mathcal{P}^i},a^i)}\Big[Q^{\pi_\theta}(s,a^{\mathcal{P}^i},a^i,\bar{a}^{-\mathcal{P}^i_+})\Big]$, $Q^{\pi_\theta,i}(s,a^{\mathcal{P}^i})= \E_{\bar{a}^{-\mathcal{P}^i}\sim\pi_{\theta}(\cdot|s,a^{\mathcal{P}^i})}\Big[Q^{\pi_\theta}(s,a^{\mathcal{P}^i},\bar{a}^{-\mathcal{P}^i})\Big]$.
\end{lemma}
Note that the policy gradient formula in Lemma \ref{BN_PG} is the same as the formula in Lemma \ref{lemma:state-based tabular softmax multi-agent policy gradient}, but with different notations. Here we uses $Q^{\pi_\theta,i}(s,a^{\mathcal{P}^i},a^i)$ instead of $Q^{\pi_\theta}(s,a^{\mathcal{P}^i_+})$ to highlight the local action $a^i$, and it is only used in the proof. They define the same quantity. We also uses $Q^{\pi_\theta,i}(s,a^{\mathcal{P}^i})$ instead of $Q^{\pi_\theta}(s,a^{\mathcal{P}^i})$ for the proof.

\begin{proof}
For agent $i$, 
$$\frac{\partial V^{\pi_\theta}(\mu)}{\partial \theta_{s,a^{\mathcal{P}^i}, a^i}^{i}}=\frac{1}{1-\gamma}\E_{\bar{s}\sim d_\mu^{\pi_\theta}}\E_{\bar{a}\sim\pi_{\theta}(\cdot|\bar{s})}\Big[A^{\pi_\theta}(\bar{s},\bar{a})\frac{\partial \log \pi_{\theta^{i}}^i(\bar{a}^i|s,a^{\mathcal{P}^i})}{\partial \theta_{s,a^{\mathcal{P}^i}, a^i}^{i}}\Big]$$
$$=\frac{1}{1-\gamma}\E_{\bar{s}\sim d_\mu^{\pi_\theta}}\E_{\bar{a}\sim\pi_{\theta}(\cdot|\bar{s})}\Big[A^{\pi_\theta}(\bar{s},\bar{a})\mathbbm{1}[\bar{s}=s]\mathbbm{1}[\bar{a}^{\mathcal{P}^i}=a^{\mathcal{P}^i}](\mathbbm{1}[\bar{a}^i=a^i]-\pi_{\theta^{i}}^i(a^i|s,a^{\mathcal{P}^i}))\Big]$$
$$=\frac{1}{1-\gamma}d_\mu^{\pi_\theta}(s)\pi_{\theta}^{\mathcal{P}^i}(a^{\mathcal{P}^i}|s)\E_{\bar{a}^i\sim\pi_{\theta^{i}}^i(\cdot|s,a^{\mathcal{P}^i}),\bar{a}^{-\mathcal{P}^i_+}\sim\pi_{\theta}(\cdot|s,a^{\mathcal{P}^i},a^i)}$$
$$\Big[A^{\pi_\theta}(s,a^{\mathcal{P}^i},\bar{a}^i,\bar{a}^{-\mathcal{P}^i_+})(\mathbbm{1}[\bar{a}^i=a^i]-\pi_{\theta^{i}}^i(a^i|s,a^{\mathcal{P}^i}))\Big]$$
$$=\frac{1}{1-\gamma}d_\mu^{\pi_\theta}(s,a^{\mathcal{P}^i})\E_{\bar{a}^i\sim\pi_{\theta^{i}}^i(\cdot|s,a^{\mathcal{P}^i}),\bar{a}^{-\mathcal{P}^i_+}\sim\pi_{\theta}(\cdot|s,a^{\mathcal{P}^i},a^i)}$$
$$\Big[A^{\pi_\theta}(s,a^{\mathcal{P}^i},\bar{a}^i,\bar{a}^{-\mathcal{P}^i_+})(\mathbbm{1}[\bar{a}^i=a^i]-\pi_{\theta^{i}}^i(a^i|s,a^{\mathcal{P}^i}))\Big]$$
$$=\frac{1}{1-\gamma}d_\mu^{\pi_\theta}(s,a^{\mathcal{P}^i})$$
$$\bigg(\E_{\bar{a}^i\sim\pi_{\theta^{i}}^i(\cdot|s,a^{\mathcal{P}^i}),\bar{a}^{-\mathcal{P}^i_+}\sim\pi_{\theta}(\cdot|s,a^{\mathcal{P}^i},a^i)}\Big[A^{\pi_\theta}(s,a^{\mathcal{P}^i},\bar{a}^i,\bar{a}^{-\mathcal{P}^i_+})\mathbbm{1}[\bar{a}^i=a^i]\Big]$$
$$-\E_{\bar{a}^{-\mathcal{P}^i}\sim\pi_{\theta}(\cdot|s,a^{\mathcal{P}^i})}\Big[A^{\pi_\theta}(s,a^{\mathcal{P}^i},\bar{a}^{-\mathcal{P}^i})\pi_{\theta^{i}}^i(a^i|s,a^{\mathcal{P}^i})\Big]\bigg)$$

$$=\frac{1}{1-\gamma}d_\mu^{\pi_\theta}(s,a^{\mathcal{P}^i})$$
$$\bigg(\pi_{\theta^{i}}^i(a^i|s,a^{\mathcal{P}^i})\E_{\bar{a}^{-\mathcal{P}^i_+}\sim\pi_{\theta}(\cdot|s,a^{\mathcal{P}^i},a^i)}\Big[A^{\pi_\theta}(s,a^{\mathcal{P}^i},a^i,\bar{a}^{-\mathcal{P}^i_+})\Big]$$
$$-\pi_{\theta^{i}}^i(a^i|s,a^{\mathcal{P}^i})\E_{\bar{a}^{-\mathcal{P}^i}\sim\pi_{\theta}(\cdot|s,a^{\mathcal{P}^i})}\Big[A^{\pi_\theta}(s,a^{\mathcal{P}^i},\bar{a}^{-\mathcal{P}^i})\Big]\bigg)$$

$$=\frac{1}{1-\gamma}d_\mu^{\pi_\theta}(s,a^{\mathcal{P}^i})\pi_{\theta^{i}}^i(a^i|s,a^{\mathcal{P}^i})$$
$$\bigg(\E_{\bar{a}^{-\mathcal{P}^i_+}\sim\pi_{\theta}(\cdot|s,a^{\mathcal{P}^i},a^i)}\Big[Q^{\pi_\theta}(s,a^{\mathcal{P}^i},a^i,\bar{a}^{-\mathcal{P}^i_+})\Big]-\E_{\bar{a}^{-\mathcal{P}^i}\sim\pi_{\theta}(\cdot|s,a^{\mathcal{P}^i})}\Big[Q^{\pi_\theta}(s,a^{\mathcal{P}^i},\bar{a}^{-\mathcal{P}^i})\Big]\bigg)$$
$$=\frac{1}{1-\gamma}d_\mu^{\pi_\theta}(s,a^{\mathcal{P}^i})\pi_{\theta^{i}}^i(a^i|s,a^{\mathcal{P}^i})\Big[Q^{\pi_\theta,i}(s,a^{\mathcal{P}^i},a^i)-Q^{\pi_\theta,i}(s,a^{\mathcal{P}^i})\Big]$$
$$=\frac{1}{1-\gamma}d_\mu^{\pi_\theta}(s,a^{\mathcal{P}^i})\pi_{\theta^{i}}^i(a^i|s,a^{\mathcal{P}^i})A^{\pi_\theta,i}(s,a^{\mathcal{P}^i},a^i)$$

\end{proof}

\begin{lemma}
\label{lemma_4}
For all agents $i$ with a round of update  $$\theta^{t+1,i}=\theta^{t,i}+\eta\nabla 
V^i_{\theta^{t,i}}(\mu)$$ with learning rates $\eta\leq \frac{(1-\gamma)^3}{8N(r_{\rm max}-r_{\rm min})}$, we have $$V^{t+1}(s)\geq V^{t}(s);Q^{t+1}(s,a)\geq Q^{t}(s,a).$$
\end{lemma}

\begin{proof}
Since $V(\theta)$ is $\frac{8N(r_{\rm max}-r_{\rm min})}{(1-\gamma)^3}$-smooth, we know that with learning rate $\eta\leq \frac{(1-\gamma)^3}{8N(r_{\rm max}-r_{\rm min})}$, $V$ is monotonic increasing and therefore $Q$ is also monotonic increasing.  
\end{proof}

\begin{lemma}
\label{lemma_5}
For all states s and actions a, there exists values $V^\infty(s) \text{ and }Q^{\infty}(s,a)$ such that as $t\rightarrow \infty, V^t(s)\rightarrow V^\infty(s),Q^t(s,a)\rightarrow Q^\infty(s,a)$.
For all agents $i$, states $s$, parent actions $a^{\mathcal{P}^i}$, local action $a^i$, there exists values $Q^{\infty,i}(s,a^{\mathcal{P}^i},a^i)\text{ and }Q^{\infty,i}(s,a^{\mathcal{P}^i})$ such that as $t\rightarrow \infty, Q^{t,i}(s,a^{\mathcal{P}^i},a^i)\rightarrow Q^{\infty,i}(s,a^{\mathcal{P}^i},a^i),Q^{t,i}(s,a^{\mathcal{P}^i})\rightarrow Q^{\infty,i}(s,a^{\mathcal{P}^i})$.
Define $$\Delta^i = \min_{\{s,a^{\mathcal{P}^i},a^i|A^{\infty,i}(s,a^{\mathcal{P}^i},a^i)\neq 0\}}|A^{\infty,i}(s,a^{\mathcal{P}^i},a^i)|.$$
$$\Delta = \min_{i}\Delta^i.$$
Further, there exists a $T_0$ such that for all $t>T_0$, agents $i$, states $s$, parent actions $a^{\mathcal{P}^i}$, local action $a^i$, $$Q^{\infty,i}(s,a^{\mathcal{P}^i},a^i)-\frac{\Delta}{4}\leq Q^{t,i}(s,a^{\mathcal{P}^i},a^i)\leq Q^{\infty,i}(s,a^{\mathcal{P}^i},a^i)+\frac{\Delta}{4}$$
\end{lemma}
\begin{proof}
$\{V^t(s)\}$ is bounded and monotonically increasing, therefore $V^t(s)\rightarrow V^\infty(s)$. Similarly, we know $Q^t(s,a)\rightarrow Q^\infty(s,a)$. 
Since the Bayesian policy is assumed to converge, we have that both $\{Q^{t,i}(s,a^{\mathcal{P}^i},a^i)\}$ and $\{Q^{t,i}(s,a^{\mathcal{P}^i})\}$ are convergent. 
For all agents $i$, states $s$, parent actions $a^{\mathcal{P}^i}$, categorize the local action $a^i$ into three groups:
$$I_0^{s,a^{\mathcal{P}^i},i}=\Bigg\{a^i|Q^{\infty,i}(s,a^{\mathcal{P}^i},a^i)=Q^{\infty,i}(s,a^{\mathcal{P}^i})\Bigg\}$$
$$I_+^{s,a^{\mathcal{P}^i},i}=\Bigg\{a^i|Q^{\infty,i}(s,a^{\mathcal{P}^i},a^i)>Q^{\infty,i}(s,a^{\mathcal{P}^i})\Bigg\}$$
$$I_-^{s,a^{\mathcal{P}^i},i}=\Bigg\{a^i|Q^{\infty,i}(s,a^{\mathcal{P}^i},a^i)<Q^{\infty,i}(s,a^{\mathcal{P}^i})\Bigg\}$$
Since $Q^{t,i}(s,a^{\mathcal{P}^i},a^i)\rightarrow Q^{\infty,i}(s,a^{\mathcal{P}^i},a^i)$ as $t\rightarrow \infty$, there exists a $T_0$ such that for all $t>T_0$, agents $i$, states $s$, parent actions $a^{\mathcal{P}^i}$, $$Q^{\infty,i}(s,a^{\mathcal{P}^i},a^i)-\frac{\Delta}{4}\leq Q^{t,i}(s,a^{\mathcal{P}^i},a^i)\leq Q^{\infty,i}(s,a^{\mathcal{P}^i},a^i)+\frac{\Delta}{4}$$
\end{proof}

\begin{lemma}
\label{lemma_6}
$\exists T_1$ such that $\forall t>T_1,i,s,a^{\mathcal{P}^i},a^i$, we have $$A^{t,i}(s,a^{\mathcal{P}^i},a^i)<-\frac{\Delta}{4}\text{ for }a^i\in I_-^{s,a^{\mathcal{P}^i},i};A^{t,i}(s,a^{\mathcal{P}^i},a^i)>\frac{\Delta}{4}\text{ for }a^i\in I_+^{s,a^{\mathcal{P}^i},i}$$
\end{lemma}

\begin{proof}
Since $\forall s, i, a^{\mathcal{P}^i}, Q^{t,i}(s,a^{\mathcal{P}^i})\rightarrow Q^{\infty,i}(s,a^{\mathcal{P}^i})$, we have that there exists $T_1>T_0$ such that for all $t>T_1$, $$Q^{\infty,i}(s,a^{\mathcal{P}^i})-\frac{\Delta}{4}\leq Q^{t,i}(s,a^{\mathcal{P}^i})\leq Q^{\infty,i}(s,a^{\mathcal{P}^i})+\frac{\Delta}{4}$$
For $a^i\in I_-^{s,a^{\mathcal{P}^i},i},$  
\begin{equation} 
\begin{split}
A^{t,i}(s,a^{\mathcal{P}^i},a^i) & = Q^{t,i}(s,a^{\mathcal{P}^i},a^i)-Q^{t,i}(s,a^{\mathcal{P}^i}) \\
 & \leq Q^{\infty,i}(s,a^{\mathcal{P}^i},a^i)+\frac{\Delta}{4}-Q^{t,i}(s,a^{\mathcal{P}^i})
 \\
 & \leq Q^{\infty,i}(s,a^{\mathcal{P}^i},a^i)+\frac{\Delta}{4}-Q^{\infty,i}(s,a^{\mathcal{P}^i})+\frac{\Delta}{4}
 \\
 & \leq -\Delta+\frac{\Delta}{4}+\frac{\Delta}{4}
 \\
 & < -\frac{\Delta}{4}
\end{split}
\end{equation}

For $a^i\in I_+^{s,a^{\mathcal{P}^i},i},$  
\begin{equation}
\begin{split}
A^{t,i}(s,a^{\mathcal{P}^i},a^i) & = Q^{t,i}(s,a^{\mathcal{P}^i},a^i)-Q^{t,i}(s,a^{\mathcal{P}^i}) \\
 & \geq Q^{\infty,i}(s,a^{\mathcal{P}^i},a^i)-\frac{\Delta}{4}-Q^{t,i}(s,a^{\mathcal{P}^i})
 \\
 & \geq Q^{\infty,i}(s,a^{\mathcal{P}^i},a^i)-\frac{\Delta}{4}-Q^{\infty,i}(s,a^{\mathcal{P}^i})-\frac{\Delta}{4}
 \\
 & \geq \Delta-\frac{\Delta}{4}-\frac{\Delta}{4}
 \\
 & > \frac{\Delta}{4}
\end{split}
\end{equation}
\end{proof}

\begin{lemma}
\label{lemma_7}
$\frac{\partial V^{t}(\mu)}{\partial \theta_{s,a^{\mathcal{P}^i}, a^i}^{i}}\rightarrow 0 \text{ as }t\rightarrow\infty$ for all agents $i$, states $s$, parent actions $a^{\mathcal{P}^i}$, local action $a^i$. This implies that $\forall i, \forall a^{\mathcal{P}^i}$, if $\lim_{t \to \infty} d_\mu^{\pi^t}(s,a^{\mathcal{P}^i})>0$, then $\forall a^i\in I_-^{s,a^{\mathcal{P}^i},i}\cup I_+^{s,a^{\mathcal{P}^i},i}, \pi^{t,i}(a^i|s,a^{\mathcal{P}^i})\rightarrow 0$ and that $\sum_{a^i\in I_0^{s,a^{\mathcal{P}^i},i}}\pi^{t,i}(a^i|s, a^{\mathcal{P}^i})\rightarrow 1$.
\end{lemma}

\begin{proof}
Since $V^{\pi_\theta}(\mu)$ is $\frac{8N(r_{\rm max}-r_{\rm min})}{(1-\gamma)^3}$-smooth, we know that with learning rate $\eta < \frac{(1-\gamma)^3}{8N(r_{\rm max}-r_{\rm min})}$, $\frac{\partial V^{t}(\mu)}{\partial \theta_{s,a^{\mathcal{P}^i}, a^i}^{i}}\rightarrow 0$ for all $i,s,a^{\mathcal{P}^i},a^i$. From lemma \ref{BN_PG} we have $$\frac{\partial V^{t}(\mu)}{\partial \theta_{s,a^{\mathcal{P}^i}, a^i}^{i}}=\frac{1}{1-\gamma}d_\mu^{\pi^t}(s,a^{\mathcal{P}^i})\pi^{t,i}(a^i|s,a^{\mathcal{P}^i})A^{t,i}(s,a^{\mathcal{P}^i},a^i)$$
\\Since from lemma \ref{lemma_6}, we know that $|A^{\pi_\theta}(s,a^{\mathcal{P}^i},a^i)|>\frac{\Delta}{4}$ for all $t>T_1$, for all $a^i\in I_-^{s,a^{\mathcal{P}^i},i}\cup I_+^{s,a^{\mathcal{P}^i},i}$
, which together with the assumption that $ \lim_{t \to \infty} d_\mu^{\pi^t}(s,a^{\mathcal{P}^i})>0$ proves $\pi^{t,i}(a^i|s,a^{\mathcal{P}^i})\rightarrow 0$. Then we also know for all $\sum_{a^i\in I_0^{s,a^{\mathcal{P}^i},i}}\pi^{t,i}(a^i|s,a^{\mathcal{P}^i})\rightarrow 1$. 
\end{proof}

From Lemma \ref{lemma_8} to Lemma \ref{lemma_14}, we prove the properties under the condition that $\lim_{t \to \infty} d_\mu^{\pi^t}(s,a^{\mathcal{P}^i})>0$, so that with Lemma \ref{lemma_7}, we know that $\pi^{t,i}(a^i|s,a^{\mathcal{P}^i})\rightarrow 0$. 

\begin{lemma}
\label{lemma_8}
For $t\geq T_1$, $\forall i, \forall a^{\mathcal{P}^i}$, if $\lim_{t \to \infty} d_\mu^{\pi^t}(s,a^{\mathcal{P}^i})>0$, then $\theta_{s,a^{\mathcal{P}^i}, a^i}^{i}$ is strictly decreasing $\forall a^i\in I_-^{s,a^{\mathcal{P}^i},i}$ and $\theta_{s,a^{\mathcal{P}^i}, a^i}^{i}$ is strictly increasing $\forall a^i\in I_+^{s,a^{\mathcal{P}^i},i}$.
\end{lemma}
\begin{proof}
From \ref{BN_PG} we have $$\frac{\partial V^{t}(\mu)}{\partial \theta_{s,a^{\mathcal{P}^i}, a^i}^{i}}=\frac{1}{1-\gamma}d_\mu^{\pi^t}(s,a^{\mathcal{P}^i})\pi^{t,i}(a^i|s,a^{\mathcal{P}^i})A^{t,i}(s,a^{\mathcal{P}^i},a^i)$$
From lemma \ref{lemma_6}, we know for all $t>T_1,a^i\in I_-^{s,a^{\mathcal{P}^i},i}, A^{t,i}(s,a^{\mathcal{P}^i},a^i) < -\frac{\Delta}{4};$ For all $a^i\in I_+^{s,a^{\mathcal{P}^i},i}, A^{t,i}(s,a^{\mathcal{P}^i},a^i) > \frac{\Delta}{4}.$ This implies that after iteration $T_1$, $\frac{\partial V^{t}(\mu)}{\partial \theta_{s,a^{\mathcal{P}^i},a^i}^{i}}<0 \forall a^i\in I_-^{s,a^{\mathcal{P}^i},i};\frac{\partial V^{t}(\mu)}{\partial \theta_{s,a^{\mathcal{P}^i},a^i}^{i}}>0 \forall a^i\in I_+^{s,a^{\mathcal{P}^i},i}.\longrightarrow$ After iteration $T_1$, $\theta_{s,a^{\mathcal{P}^i},a^i}^{i}$ is strictly decreasing $\forall a^i\in I_-^{s,a^{\mathcal{P}^i},i}$ and $\theta_{s,a^{\mathcal{P}^i},a^i}^{i}$ is strictly increasing $\forall a^i\in I_+^{s,a^{\mathcal{P}^i},i}$.
\end{proof}

\begin{lemma}
\label{lemma_9}
For all $i,s,a^{\mathcal{P}^i},a^i$, if $\lim_{t \to \infty} d_\mu^{\pi^t}(s,a^{\mathcal{P}^i})>0$ and $I_+^{s,a^{\mathcal{P}^i},i}\neq \emptyset$, then we have:
$$\max_{a^i\in I_0^{s,a^{\mathcal{P}^i},i}}\theta_{s,a^{\mathcal{P}^i},a^i}^{t,i}\rightarrow\infty, \text{ min}_{a^i\in \mathbb{A}^i}\theta_{s,a^{\mathcal{P}^i},a^i}^{t,i}\rightarrow-\infty$$
\end{lemma}

\begin{proof}
Since $I_+^{s,a^{\mathcal{P}^i},i}\neq \emptyset$, we have some action $a_+^i\in I_+^{s,a^{\mathcal{P}^i},i}$. From lemma \ref{eq:7}, we know $$\pi^{t,i}(a^i_+|s,a^{\mathcal{P}^i})\rightarrow 0 \text{ as }t\rightarrow \infty$$ 
$$\longrightarrow \frac{\text{exp}(\theta_{s,a^{\mathcal{P}^i},a_+^i}^{t,i})}{\sum_{a^i\in \mathbb{A}^i}\text{exp}(\theta_{s,a^{\mathcal{P}^i},a^i}^{t,i})}\rightarrow 0 \text{ as }t\rightarrow \infty$$
From lemma \ref{lemma_8} we know $\theta_{s,a^{\mathcal{P}^i},a_+^i}^{t,i}$ is monotonically increasing, which implies 
$$\sum_{a^i\in \mathbb{A}^i}\text{exp}(\theta_{s,a^{\mathcal{P}^i},a^i}^{t,i})\rightarrow \infty \text{ as }t\rightarrow \infty$$
From lemma \ref{lemma_7}, we also know $$\sum_{a^i\in I_0^{s,a^{\mathcal{P}^i},i}}\pi^{t,i}(a^i|s,a^{\mathcal{P}^i})\rightarrow 1$$
$$\longrightarrow \frac{\sum_{a^i\in I_0^{s,a^{\mathcal{P}^i},i}}\text{exp}(\theta_{s,a^{\mathcal{P}^i},a^i}^{t,i})}{\sum_{a^i\in \mathbb{A}^i}\text{exp}(\theta_{s,a^{\mathcal{P}^i},a^i}^{t,i})}\rightarrow 1$$
Since denominator does to $\infty$, we know
$$\sum_{a^i\in I_0^{s,a^{\mathcal{P}^i},i}}\text{exp}(\theta_{s,a^{\mathcal{P}^i},a^i}^{t,i})\rightarrow\infty$$
which implies $$\max_{a^i\in I_0^{s,a^{\mathcal{P}^i},i}}\theta_{s,a^{\mathcal{P}^i},a^i}^{t,i}\rightarrow\infty$$
Note this also implies $\max_{a^i\in \mathbb{A}^i}\theta_{s,a^{\mathcal{P}^i},a^i}^{t,i}\rightarrow\infty$. The sum of the gradient is always zero: $\sum_{a^i\in \mathbb{A}^i}\frac{\partial V^{t}(\mu)}{\partial \theta_{s,a^{\mathcal{P}^i},a^i}^{i}}=\frac{1}{1-\gamma}d_\mu^{\pi^t}(s)\pi^{t,\mathcal{P}^i}(a^{\mathcal{P}^i}|s)\sum_{a^i\in \mathbb{A}^i}\pi^{t,i}(a^i|s,a^{\mathcal{P}^i})A^{t,i}(s,a^{\mathcal{P}^i},a^i)=0$. Thus, $\sum_{a^i\in \mathbb{A}^i}\theta_{s,a^{\mathcal{P}^i},a^i}^{t,i}=\sum_{a^i\in \mathbb{A}^i}\theta_{s,a^{\mathcal{P}^i},a^i}^{0,i}$ which is a constant. Since $\max_{a^i\in \mathbb{A}^i}\theta_{s,a^{\mathcal{P}^i},a^i}^{t,i}\rightarrow\infty$, we know $$\min_{a^i\in \mathbb{A}^i}\theta_{s,a^{\mathcal{P}^i},a^i}^{t,i}\rightarrow-\infty$$
\end{proof}

\begin{lemma}
\label{lemma_10}
For some $s,i,a^{\mathcal{P}^i}$, suppose $a_+^i\in I_+^{s,a^{\mathcal{P}^i},i}$. $\forall a\in I_0^{s,a^{\mathcal{P}^i},i},$ if $\exists t\geq T_1$ such that $\pi^{t,i}(a|s,a^{\mathcal{P}^i})\leq \pi^{t,i}(a_+^i|s,a^{\mathcal{P}^i})$, then $\forall \tau\geq t, \pi^{\tau,i}(a|s,a^{\mathcal{P}^i})\leq \pi^{\tau,i}(a_+^i|s,a^{\mathcal{P}^i})$.  

\end{lemma}

\begin{proof}
Suppose $a_+^i\in I_+^{s,a^{\mathcal{P}^i},i}, a\in I_0^{s,a^{\mathcal{P}^i},i},$ if $\pi^{t,i}(a|s,a^{\mathcal{P}^i})\leq \pi^{t,i}(a_+^i|s,a^{\mathcal{P}^i})$, then
$$\frac{\partial V^{t}(\mu)}{\partial \theta_{s,a^{\mathcal{P}^i}, a^i}^{i}}=\frac{1}{1-\gamma}d_\mu^{\pi^t}(s,a^{\mathcal{P}^i})\pi^{t,i}(a^i|s,a^{\mathcal{P}^i})A^{t,i}(s,a^{\mathcal{P}^i},a^i)$$
$$\frac{\partial V^{t}(\mu)}{\partial \theta_{s,a^{\mathcal{P}^i},a}^{i}}=\frac{1}{1-\gamma}d_\mu^{\pi^t}(s,a^{\mathcal{P}^i})\pi^{t,i}(a^i|s,a^{\mathcal{P}^i})\Big[Q^{t,i}(s,a^{\mathcal{P}^i},a^i)-Q^{t,i}(s,a^{\mathcal{P}^i})\Big]$$
$$\leq \frac{1}{1-\gamma}d_\mu^{\pi_\theta}(s,a^{\mathcal{P}^i})\pi^{t,i}(a_+^i|s,a^{\mathcal{P}^i})\Big[Q^{t,i}(s,a^{\mathcal{P}^i},a_+^i)-Q^{t,i}(s,a^{\mathcal{P}^i})\Big]=\frac{\partial V^{t}(\mu)}{\partial \theta_{s,a^{\mathcal{P}^i},a_+^i}^{i}}$$
where the last step holds because $Q^{{t,i}}(s,a^{\mathcal{P}^i},a_+^i)\geq Q^{{\infty,i}}(s,a^{\mathcal{P}^i},a_+^i)-\frac{\Delta}{4}\geq Q^{{\infty,i}}(s,a^{\mathcal{P}^i},a)+\Delta-\frac{\Delta}{4}\geq Q^{{t,i}}(s,a^{\mathcal{P}^i},a)-\frac{\Delta}{4}+\Delta-\frac{\Delta}{4}> Q^{{t,i}}(s,a^{\mathcal{P}^i},a)$
for $t>T_0$.
\\We can then partition $I_0^{s,a^{\mathcal{P}^i},i}$ into $B_0^{s,a^{\mathcal{P}^i},i}(a_+^i)$ and $\Bar{B}_0^{s,a^{\mathcal{P}^i},i}(a_+^i)$ as follows: 
$$B_0^{s,a^{\mathcal{P}^i},i}(a_+^i):\{a|a\in I_0^{s,a^{\mathcal{P}^i},i}\text{ and }\forall t\geq T_0,\pi^{t,i}(a_+^i|s,\mathcal{P}(i))< \pi^{t,i}(a|s,\mathcal{P}(i))\}$$
$$\Bar{B}_0^{s,a^{\mathcal{P}^i},i}(a_+^i):I_0^{s,a^{\mathcal{P}^i},i}\setminus B_0^{s,\mathcal{P}(i),i}(a_+^i).$$
\end{proof}

\begin{lemma}
\label{lemma_11}
For some $s,i,a^{\mathcal{P}^i}$, if $\lim_{t \to \infty} d_\mu^{\pi^t}(s,a^{\mathcal{P}^i})>0$, then suppose $I_+^{s,a^{\mathcal{P}^i},i}\neq\emptyset$. $\forall a_+^i\in I_+^{s,a^{\mathcal{P}^i},i}$, we have that $B_0^{s,a^{\mathcal{P}^i},i}(a_+^i)\neq\emptyset$ and that $$\sum_{a^i\in B_0^{s,a^{\mathcal{P}^i},i}(a_+^i)}\pi^{t,i}(a^i|s,a^{\mathcal{P}^i})\rightarrow 1\text{, as } t\rightarrow\infty.$$
This implies that: $$\max_{a^i\in B_0^{s,a^{\mathcal{P}^i},i}(a_+^i)}\theta_{s,a^{\mathcal{P}^i},a^i}^{t,i}\rightarrow\infty.$$

\end{lemma}

\begin{proof}
Let $a_+^i\in I_+^{s,a^{\mathcal{P}^i},i}$. Consider any $\Bar{a}^i\in \Bar{B}_0^{s,a^{\mathcal{P}^i},i}(a_+^i)$. Then by definition of $\Bar{B}_0^{s,a^{\mathcal{P}^i},i}(a_+^i)$, there exists $t^\prime>T_0$ such that $\pi^{t^\prime,i}(a_+^i|s,a^{\mathcal{P}^i})\geq \pi^{t^\prime,i}(\Bar{a}^i|s,a^{\mathcal{P}^i})$. From lemma \ref{lemma_10}, we know $\forall \tau>t^\prime,\pi^{\tau,i}(a_+^i|s,a^{\mathcal{P}^i})\geq \pi^{\tau,i}(\Bar{a}^i|s,a^{\mathcal{P}^i})$. From lemma \ref{lemma_7}, we know $\pi^{t,i}(a_+^i|s,a^{\mathcal{P}^i})\rightarrow 0 \text{ as }t\rightarrow \infty$, which implies $$\pi^{t,i}(\Bar{a}^i|s,a^{\mathcal{P}^i})\rightarrow 0 \text{ as }t\rightarrow \infty.$$ \\Since $B_0^{s,a^{\mathcal{P}^i},i}(a_+^i)\cup \Bar{B}_0^{s,a^{\mathcal{P}^i},i}(a_+^i) = I_0^{s,i}$ and $\sum_{a^i\in I_0^{s,a^{\mathcal{P}^i},i}}\pi^{t,i}(a^i|s,a^{\mathcal{P}^i})\rightarrow 1$, we know $$\sum_{a^i\in B_0^{s,a^{\mathcal{P}^i},i}(a_+^i)}\pi^{t,i}(a^i|s,a^{\mathcal{P}^i})\rightarrow 1$$
$$B_0^{s,a^{\mathcal{P}^i},i}(a_+^i)\neq \oldemptyset$$
Using the same techniques in \ref{lemma_9}, we know $$\max_{a^i\in B_0^{s,a^{\mathcal{P}^i},i}(a_+^i)}\theta_{s,a^{\mathcal{P}^i},a^i}^{t,i}\rightarrow\infty$$
\end{proof}

\begin{lemma}
\label{lemma_12}
Consider any $s,a^{\mathcal{P}^i}$, where $I_+^{s,a^{\mathcal{P}^i},i}\neq\emptyset$. Then, $\forall a_+^i\in I_+^{s,a^{\mathcal{P}^i},i}, \exists T_{a^{\mathcal{P}^i},a_+^i}$ such that $\forall t>T_{a^{\mathcal{P}^i},a_+^i},\forall a^i\in\Bar{B}_0^{s,i}(a_+^i),$
$$\pi^{t,i}(a_+^i|s,a^{\mathcal{P}^i}))> \pi^{t,i}(a^i|s,a^{\mathcal{P}^i}))$$
\end{lemma}

\begin{proof}
By the definition of $\Bar{B}_0^{s,i}(a_+^i)$ and lemma \ref{lemma_10}, $\forall a^i\in \Bar{B}_0^{s,i}(a_+^i)$, there exists $t_{a^{\mathcal{P}^i},a^i}>T_0$ such that $\forall \tau>t_{a^{\mathcal{P}^i},a^i}$, $\pi^{\tau,i}(a_+^i|s,a^{\mathcal{P}^i})> \pi^{\tau,i}(a^i|s,a^{\mathcal{P}^i})$. We can choose $T_{a^{\mathcal{P}^i},a_+^i}=\max_{a^{\mathcal{P}^i},a^i\in B_0^{s,i}(a_+^i)}t_{a^i}$.
\end{proof}

\begin{lemma}
\label{lemma_13}
$\forall i, a^{\mathcal{P}^i}, a^i$, if $\lim_{t \to \infty} d_\mu^{\pi^t}(s,a^{\mathcal{P}^i})>0$, then we have $\forall a^i_+\in I_+^{s,a^{\mathcal{P}^i},i}$, $\theta_{s,a^{\mathcal{P}^i},a^i_+}^{i}$ is lower bounded as $t\rightarrow\infty$ and $\forall a^i_-\in I_-^{s,a^{\mathcal{P}^i},i}$, $\theta_{s,a^{\mathcal{P}^i},a^i_-}^{i}\rightarrow-\infty$ as $t\rightarrow\infty$. 
\end{lemma}

\begin{proof}
From lemma \ref{lemma_8}, we know that $\forall a^i_+\in I_+^{s,a^{\mathcal{P}^i},i}$, after $T_1$, $\theta_{s,a^{\mathcal{P}^i},a^i_+}^{i}$ is strictly increasing, and is therefore bounded from below.
\\For the second claim, we know from lemma \ref{lemma_8} that $\forall a^i_-\in I_-^{s,a^{\mathcal{P}^i},i}$, after $T_1$, 
$\theta_{s,a^{\mathcal{P}^i},a^i_-}^{i}$ is strictly decreasing. Then, by monotone convergence theorem, we know $\text{lim}_{t\rightarrow \infty}\theta_{s,a^{\mathcal{P}^i},a^i_-}^{i}$ exists and is either $-\infty$ or some constant $\theta_{0}^{i}$. We now prove by contraction that $\text{lim}_{t\rightarrow \infty}\theta_{s,a^{\mathcal{P}^i},a^i_-}^{i}$ cannot be some constant $\theta_{0}^{i}$. Suppose $\text{lim}_{t\rightarrow \infty}\theta_{s,a^{\mathcal{P}^i},a^i_-}^{i}=\theta_{0}^{i}$. We immediately know that $\forall t\geq T_1,\theta_{s,a^{\mathcal{P}^i},a^i_-}^{i}>\theta_{0}^{i}$. By lemma \ref{lemma_9}, we know $\exists a^i \in \mathbb{A}^i$ such that \begin{equation}\label{eq:3}\text{lim }\underset{t\rightarrow \infty}{\text{inf}}\theta_{s,a^{\mathcal{P}^i},a^i}^{t,i}=-\infty\end{equation}
Let us consider some $\delta^i>0$ such that $\theta_{s,a^{\mathcal{P}^i},a^i}^{T_1,i}\geq \theta_{0}^{i}-\delta^i$. Now for $t\geq T_1$, define $\tau^i(t)$ to be the largest iteration in $[T_1,t]$ such that $\theta_{s,a^{\mathcal{P}^i},a^i}^{\tau^i(t),i}\geq \theta_{0}^{i}-\delta^i$.
Define $\Tau^{t,i}$ to be subsequence $\{t^{\prime}\}$ of the interval $(\tau^i(t),t)$ such that $\theta_{s,a^{\mathcal{P}^i},a^i}^{t^\prime,i}$ decreases. 
\\Define 
$$Z^{t,a^{\mathcal{P}^i},i}=\sum_{t^\prime\in\Tau^{t,i}}\frac{\partial V^{t^\prime}(\mu)}{\partial \theta_{s,a^{\mathcal{P}^i},a^i}^{i}}$$
where $Z^{t,a^{\mathcal{P}^i},i}=0$ if $\Tau^{t,i}=\emptyset$.
\\For non-empty $\Tau^{t,i}$, we have:
$$Z^{t,a^{\mathcal{P}^i},i}=\sum_{t^\prime\in\Tau^{t,i}}\frac{\partial V^{t^\prime}(\mu)}{\partial \theta_{s,a^{\mathcal{P}^i},a^i}^{i}}\leq \sum_{t^\prime=\tau^i(t)+1}^{t-1}\frac{\partial V^{t^\prime}(\mu)}{\partial \theta_{s,a^{\mathcal{P}^i},a^i}^{i}}\leq \sum_{t^\prime=\tau^i(t)}^{t-1}\frac{\partial V^{t^\prime}(\mu)}{\partial \theta_{s,a^{\mathcal{P}^i},a^i}^{i}}+\frac{1}{(1-\gamma)}(V_{\rm max}-V_{\rm min})$$
$$=\frac{1}{\eta}(\theta_{s,a^{\mathcal{P}^i},a^i}^{t,i}-\theta_{s,a^{\mathcal{P}^i},a^i}^{\tau^i(t),i})+\frac{1}{(1-\gamma)}(V_{\rm max}-V_{\rm min})$$
where we have used that $|\frac{\partial V^{t^\prime}(\mu)}{\partial \theta_{s,a^{\mathcal{P}^i},a^i}^{i}}|\leq \frac{1}{(1-\gamma)}(V_{\rm max}-V_{\rm min})$. \\By equation $(\ref{eq:3})$, we know \begin{equation} \label{eq:4}\text{lim }\underset{t\rightarrow \infty}{\text{inf}}Z^{t,a^{\mathcal{P}^i},i}=-\infty\end{equation}
$$\frac{\partial V^{t}(\mu)}{\partial \theta_{s,a^{\mathcal{P}^i}, a^i}^{i}}=\frac{1}{1-\gamma}d_\mu^{\pi_\theta}(s)\pi^{t,i}(a^i|s,a^{\mathcal{P}^i})\pi^{t,\mathcal{P}^i}(a^{\mathcal{P}^i}|s)A^{t,i}(s,a^{\mathcal{P}^i},a^i)$$
\\For any $\Tau^{t,i}\neq \emptyset,\forall t^{\prime}\in \Tau^{t,i}$, from lemma \ref{BN_PG}, we know:
$$\Bigg|\frac{\partial V^{t^\prime}(\mu)/\partial \theta_{s,a^i_-}^{i}}{\partial V^{t^\prime}(\mu)/\partial \theta_{s,a^i}^{i}}\Bigg|=\Bigg|\frac{\pi^{t^\prime,i}(a^i_-|s,a^{\mathcal{P}^i})A^{t^\prime,i}(s,a^{\mathcal{P}^i},a^i_-)}{\pi^{t^\prime,i}(a^i|s,a^{\mathcal{P}^i})A^{t^\prime,i}(s,a^{\mathcal{P}^i},a^i)}\Bigg|\geq \text{exp}(\theta_{0}^{i}-\theta_{s,a^{\mathcal{P}^i},a^i}^{t^\prime,i})\frac{\Delta}{4(V_{\rm max}-V_{\rm min})}$$
$$\geq \text{exp}(\delta^i)\frac{\Delta}{4(V_{\rm max}-V_{\rm min})}$$
where we have used that $|A^{t^\prime,i}(s,a^{\mathcal{P}^i},a^i)|\leq V_{\rm max}-V_{\rm min}$ and $\forall t^\prime>T_1,|A^{t^\prime,i}(s,a^{\mathcal{P}^i},a^i_-)|\geq \frac{\Delta}{4}$.
\\Since both $\frac{\partial V^{t^\prime}(\mu)}{\partial \theta_{s,a^{\mathcal{P}^i},a^i_-}^{i}}$ and $\frac{\partial V^{t^\prime}(\mu)}{\partial \theta_{s,a^{\mathcal{P}^i},a^i}^{i}}$ are negative, we can get: \begin{equation} \label{eq:5}
\frac{\partial V^{t^\prime}(\mu)}{\partial \theta_{s,a^{\mathcal{P}^i},a^i_-}^{i}}\leq \text{exp}(\delta^i)\frac{\Delta}{4(V_{\rm max}-V_{\rm min})}\frac{\partial V^{t^\prime}(\mu)}{\partial \theta_{s,a^{\mathcal{P}^i},a^i}^{i}} 
\end{equation}
For non-empty $\Tau^{t,i}$, 
$$\frac{1}{\eta}(\theta_{s,a^{\mathcal{P}^i},a^i_-}^{t,i}-\theta_{s,a^{\mathcal{P}^i},a^i_-}^{T_1,i})=\sum_{t^\prime=T_1}^{t-1}\frac{\partial V^{t^\prime}(\mu)}{\partial \theta_{s,a^{\mathcal{P}^i},a^i_-}^{i}}\leq \sum_{t^\prime\in\Tau^{t,i}}\frac{\partial V^{t^\prime}(\mu)}{\partial \theta_{s,a^{\mathcal{P}^i},a^i_-}^{i}}$$
By Equation (\ref{eq:5})
$$\leq \text{exp}(\delta^i)\frac{\Delta}{4(V_{\rm max}-V_{\rm min})}\sum_{t^\prime\in\Tau^{t,i}}\frac{\partial V^{t^\prime}(\mu)}{\partial \theta_{s,a^{\mathcal{P}^i},a^i}^{i}}$$
$$=\text{exp}(\delta^i)\frac{\Delta}{4(V_{\rm max}-V_{\rm min})}Z^{t,a^{\mathcal{P}^i},i}$$
which together with the fact that $\theta_{s,a^{\mathcal{P}^i},a^i_-}^{T_1,i}$ is some finite constant and equation ($\ref{eq:4}$) lead to $$\theta_{s,a^{\mathcal{P}^i},a^i_-}^{t,i}\rightarrow -\infty \text{ as }t\rightarrow\infty$$
this contradicts the assumption that $\{\theta_{s,a^{\mathcal{P}^i},a_-^{t,i}}^i\}_{t\geq T_1}$ is lower bounded by $\theta_0^i$ and complete the proof.
\end{proof}

\begin{lemma}
\label{lemma_14}
Consider any $s,a^{\mathcal{P}^i}$ where $I_+^{s,a^{\mathcal{P}^i},i}\neq \emptyset$. Then, if $\lim_{t \to \infty} d_\mu^{\pi^t}(s,a^{\mathcal{P}^i})>0$, we have $\forall a_+^i\in I_+^{s,a^{\mathcal{P}^i},i}$, $$\sum_{a^i\in B_0^{s,a^{\mathcal{P}^i},i}(a_+^i)}\theta_{s,a^{\mathcal{P}^i},a^i}^{t,i}\rightarrow\infty $$

\end{lemma}

\begin{proof}
For any $a^i\in B_0^{s,a^{\mathcal{P}^i},i}(a_+^i)$. By definition, we know that $\forall t>T_0, \pi^{t,i}(a_+^i|s,a^{\mathcal{P}^i})< \pi^{t,i}(a^i|s,a^{\mathcal{P}^i})$, which implies that $\theta_{s,a^{\mathcal{P}^i},a^i_+}^{t,i}<\theta_{s,a^{\mathcal{P}^i},a^i}^{t,i}$. Since in lemma \ref{lemma_13}, $\theta_{s,a^{\mathcal{P}^i},a^i_+}^{t,i}$ is lower bounded as $t\rightarrow\infty$, we know that $\theta_{s,a^{\mathcal{P}^i},a^i}^{t,i}$ is lower bounded as $t\rightarrow\infty$. This together with lemma \ref{lemma_11} proves that $$\sum_{a^i\in B_0^{s,a^{\mathcal{P}^i},i}(a_+^i)}\theta_{s.a^{\mathcal{P}^i},a^i}^{t,i}\rightarrow\infty $$
\end{proof}

\begin{lemma}
\label{lemma_15}
$\forall i, a^{\mathcal{P}^i}, a^i$, if $\lim_{t \to \infty} d_\mu^{\pi^t}(s,a^{\mathcal{P}^i})>0$, then $I_+^{s,a^{\mathcal{P}^i},i}=\emptyset$.    
\end{lemma}

\begin{proof} 
Suppose $I_+^{s,a^{\mathcal{P}^i},i}$ is non-empty for some $s,i,a^{\mathcal{P}^i}$, else the proof is complete. Let $a_+^i\in I_+^{s,a^{\mathcal{P}^i},i}$. Then, by lemma \ref{lemma_14}, we know
\begin{equation} \label{eq:6}
\sum_{a^i\in B_0^{s,a^{\mathcal{P}^i},i}(a_+^i)}\theta_{s,a^{\mathcal{P}^i},a^i}^{t,i}\rightarrow\infty
\end{equation}
For $a^i\in I_-^{s,a^{\mathcal{P}^i},i}$, since $\frac{\pi^{t,i}(a^i|s,a^{\mathcal{P}^i})}{\pi^{t,i}(a_+^i|s,a^{\mathcal{P}^i})}=\text{exp}(\theta_{s,a^{\mathcal{P}^i},a^i}^{t,i}-\theta_{s,a^{\mathcal{P}^i},a^i_+}^{t,i})\rightarrow 0$ (as $\theta_{s,a^{\mathcal{P}^i},a^i_+}^{t,i}$ is lower bounded and $\theta_{s,a^{\mathcal{P}^i},a^i}^{t,i}\rightarrow-\infty$ by lemma \ref{lemma_13}), there exists $T_2>T_0$ such that $$\frac{\pi^{t,i}(a^i|s,a^{\mathcal{P}^i})}{\pi^{t,i}(a_+^i|s,a^{\mathcal{P}^i})}<\frac{\Delta}{8|\mathcal{A}^i|(V_{\rm max}-V_{\rm min})}$$
\begin{equation} \label{eq:7}
\longrightarrow -\sum_{a^i\in I_-^{s,a^{\mathcal{P}^i},i}}\pi^{t,i}(a^i|s,a^{\mathcal{P}^i})(V_{\rm max}-V_{\rm min})>-\pi^{t,i}(a_+^i|s,a^{\mathcal{P}^i})\frac{\Delta}{8}
\end{equation}
For $a^i\in \Bar{B}_0^{s,a^{\mathcal{P}^i},i}(a^i_+)$, by definition of $\Bar{B}_0^{s,a^{\mathcal{P}^i},i}(a^i_+)$, we have $A^{t,i}(s,a^{\mathcal{P}^i},a^i)\rightarrow 0$ and by lemma \ref{lemma_12}, $\forall t>T_{a^{\mathcal{P}^i},a_+^i} 1<\frac{\pi^{t,i}(a_+^i|s,a^{\mathcal{P}^i})}{\pi^{t,i}(a^i|s,a^{\mathcal{P}^i})}$
. Then, $\exists T_3>T_2,T_{a^{\mathcal{P}^i},a_+^i}$ such that 
$$|A^{t,i}(s,a^{\mathcal{P}^i},a^i)|<\frac{\pi^{t,i}(a_+^i|s,a^{\mathcal{P}^i})}{\pi^{t,i}(a^i|s,a^{\mathcal{P}^i})}\frac{\Delta}{16|\mathcal{A}^i|}$$
$$\longrightarrow\sum_{a^i\in \Bar{B}_0^{s,a^{\mathcal{P}^i},i}(a_+^i)}\pi^{t,i}(a^i|s,a^{\mathcal{P}^i})|A^{t,i}(s,a^{\mathcal{P}^i},a^i)|<\pi^{t,i}(a^i_+|s,a^{\mathcal{P}^i})\frac{\Delta}{16}$$
\begin{equation} \label{eq:8}
\longrightarrow -\pi^{t,i}(a^i_+|s,a^{\mathcal{P}^i})\frac{\Delta}{16}<\sum_{a^i\in \Bar{B}_0^{s,a^{\mathcal{P}^i},i}(a_+^i)}\pi^{t,i}(a^i|s,a^{\mathcal{P}^i})A^{t,i}(s,a^{\mathcal{P}^i},a^i)<\pi^{t,i}(a^i_+|s,a^{\mathcal{P}^i})\frac{\Delta}{16}
\end{equation}
For $t>T_3$, 
$$0=\sum_{a^i\in\mathcal{A}^i}\pi^{t,i}(a^i|s,a^{\mathcal{P}^i})A^{t,i}(s,a^{\mathcal{P}^i},a^i)$$
$$=\sum_{a^i\in I_0^{s,a^{\mathcal{P}^i},i}}\pi^{t,i}(a^i|s,a^{\mathcal{P}^i})A^{t,i}(s,a^{\mathcal{P}^i},a^i)+\sum_{a^i\in I_+^{s,a^{\mathcal{P}^i},i}}\pi^{t,i}(a^i|s,a^{\mathcal{P}^i})A^{t,i}(s,a^{\mathcal{P}^i},a^i)$$ $$+\sum_{a^i\in I_-^{s,a^{\mathcal{P}^i},i}}\pi^{t,i}(a^i|s,a^{\mathcal{P}^i})A^{t,i}(s,a^{\mathcal{P}^i},a^i)$$
$$\stackrel{(a)}{\geq}\sum_{a^i\in B_0^{s,a^{\mathcal{P}^i},i}(a_+^i)}\pi^{t,i}(a^i|s,a^{\mathcal{P}^i})A^{t,i}(s,a^{\mathcal{P}^i},a^i)+\sum_{a^i\in \Bar{B}_0^{s,a^{\mathcal{P}^i},i}(a_+^i)}\pi^{t,i}(a^i|s,a^{\mathcal{P}^i})A^{t,i}(s,a^{\mathcal{P}^i},a^i)$$
$$+\pi^{t,i}(a^i_+|s,a^{\mathcal{P}^i})A^{t,i}(s,a^{\mathcal{P}^i},a^i_+)+\sum_{a^i\in I_-^{s,a^{\mathcal{P}^i},i}}\pi^{t,i}(a^i|s,a^{\mathcal{P}^i})A^{t,i}(s,a^{\mathcal{P}^i},a^i)$$
$$\stackrel{(b)}{\geq}\sum_{a^i\in B_0^{s,a^{\mathcal{P}^i},i}(a_+^i)}\pi^{t,i}(a^i|s,a^{\mathcal{P}^i})A^{t,i}(s,a^{\mathcal{P}^i},a^i)+\sum_{a^i\in \Bar{B}_0^{s,a^{\mathcal{P}^i},i}(a_+^i)}\pi^{t,i}(a^i|s,a^{\mathcal{P}^i})A^{t,i}(s,a^{\mathcal{P}^i},a^i)+\pi^{t,i}(a^i_+|s,a^{\mathcal{P}^i})\frac{\Delta}{4}$$
$$-\sum_{a^i\in I_-^{s,a^{\mathcal{P}^i},i}}\pi^{t,i}(a^i|s,a^{\mathcal{P}^i})(V_{\rm max}-V_{\rm min})$$
$$\stackrel{(c)}{\geq}\sum_{a^i\in B_0^{s,a^{\mathcal{P}^i},i}(a_+^i)}\pi^{t,i}(a^i|s,a^{\mathcal{P}^i})A^{t,i}(s,a^{\mathcal{P}^i},a^i)-\pi^{t,i}(a^i_+|s,a^{\mathcal{P}^i})\frac{\Delta}{16}+\pi^{t,i}(a^i_+|s,a^{\mathcal{P}^i})\frac{\Delta}{4}-\pi^{t,i}(a_+^i|s,a^{\mathcal{P}^i})\frac{\Delta}{8}$$
$$>\sum_{a^i\in B_0^{s,a^{\mathcal{P}^i},i}(a_+^i)}\pi^{t,i}(a^i|s,a^{\mathcal{P}^i})A^{t,i}(s,a^{\mathcal{P}^i},a^i)$$
where (a) uses $\forall a^i\in I_+^{s,a^{\mathcal{P}^i},i} \text{ and }t>T_3>T_1, A^{t,i}(s,a^{\mathcal{P}^i},a^i)>0$ from lemma \ref{lemma_6}, (b) uses $\forall t>T_3>T_1, A^{t,i}(s,a^{\mathcal{P}^i},a^i_+)>\frac{\Delta}{4}$ from lemma \ref{lemma_6} and $A^{t,i}(s,a^{\mathcal{P}^i},a^i)\geq -(V_{\rm max}-V_{\rm min})$, (c) uses equation $(\ref{eq:7})$ and equation $(\ref{eq:8})$. This implies that $$\forall t>T_3, \sum_{a^i\in B_0^{s,a^{\mathcal{P}^i},i}(a_+^i)}\frac{\partial V^{t}(\mu)}{\partial \theta_{s,a^{\mathcal{P}^i},a^i}^{i}}<0$$
which contradicts with equation $(\ref{eq:6})$ which leads to $$\text{lim}_{t\rightarrow\infty}\sum_{a^i\in B_0^{s,a^{\mathcal{P}^i},i}(a_+^i)}(\theta_{s,a^{\mathcal{P}^i},a^i}^{t,i}-\theta_{s,a^{\mathcal{P}^i},a^i}^{T_3,i})=\eta\sum_{t=T_3}^\infty\sum_{a^i\in B_0^{s,a^{\mathcal{P}^i},i}(a_+^i)}\frac{\partial V^{t}(\mu)}{\partial \theta_{s,a^{\mathcal{P}^i},a^i}^{i}}\rightarrow\infty$$
Therefore, the set $I_+^{s,a^{\mathcal{P}^i},i}=\emptyset$.
\end{proof}

\begin{theorem}
\label{proof:Nash}
Under Assumptions \ref{assumption:discounted state visitation distribution} - \ref{assumption:PG},
suppose every agent $i$ follows the policy gradient dynamics \eqref{eq:PG}, which results in the update dynamics \eqref{eq:BayesianPG} for each each agent $i$, parent actions $a^{\mathcal{P}^i}$, and local action $a^i$, with $\eta\leq \frac{(1-\gamma)^3}{8N(r_{\rm max}-r_{\rm min})}$, then the converged BN policy 
$(\pi_{\theta^1_*}^1, \cdots, \pi_{\theta^N_*}^N,\mathcal{G})$ is a Nash policy.
\end{theorem}

\begin{proof}
For convenience, denote  $\sum_{a^{-\mathcal{P}^i}}\pi_{\theta}(a^{-\mathcal{P}^i},a^{\mathcal{P}^i}|s)\text{ as }\overline{\pi}^{\mathcal{P}^i}_{\theta}(\cdot|s)$ so that $d_\mu^{\pi_\theta}(s,a^{\mathcal{P}^i}) = d_\mu^{\pi_\theta}(s)\overline{\pi}^{\mathcal{P}^i}_{\theta}(\cdot|s)$.
\\$\forall i\in\mathcal{N}$, let $\theta^\prime_* = [\theta^{-i}_*,\tilde{\theta}^i_*]$ be the parameters of any joint policy where only agent $i$'s parameters are changed.
\\By performance difference lemma, 
$$V^{\pi_{\theta^\prime_*}}-V^{\pi_{\theta_*}}=\frac{1}{1-\gamma}\E_{\bar{s}\sim d_\mu^{\pi_{\theta^\prime_*}}}\E_{\bar{a}\sim\pi_{{\theta^\prime_*}}}\Big[A^{\pi_{\theta_*}}(\bar{s},\bar{a})\Big]$$
$$=\frac{1}{1-\gamma}\E_{\bar{s}\sim d_\mu^{{\pi_{\theta^\prime_*}}}(\cdot)}\E_{\bar{a}^{\mathcal{P}^i}\sim\overline{\pi}^{\mathcal{P}^i}_{\theta^\prime_*}(\cdot|\bar{s})}\E_{\bar{a}^{i}\sim\pi^{i}_{\tilde{\theta}^i_*}(\cdot|\bar{s},\bar{a}^{\mathcal{P}^i})}\E_{\bar{a}^{-\mathcal{P}^i_+}\sim\pi^{-\mathcal{P}^i_+}_{\theta_*^\prime}(\cdot|\bar{s},a^{\mathcal{P}^i_+})}\Big[Q^{\pi_{\theta_*}}(\bar{s},\bar{a}^{\mathcal{P}^i},\bar{a}^i,\bar{a}^{-\mathcal{P}^i_+})-V^{\pi_{\theta_*}}(\bar{s})\Big]$$
Since $(\theta^\prime_*)^{-i}=\theta_*^{-i}$ which means $\overline{\pi}^{\mathcal{P}^i}_{\theta_*^\prime}(\cdot|\bar{s})=\overline{\pi}^{\mathcal{P}^i}_{\theta_*}(\cdot|\bar{s}), \pi^{-\mathcal{P}^i_+}_{\theta_*^\prime}(\cdot|\bar{s},a^{\mathcal{P}^i_+})=\pi^{-\mathcal{P}^i_+}_{\theta_*}(\cdot|\bar{s},a^{\mathcal{P}^i_+}) $,
$$=\frac{1}{1-\gamma}\E_{\bar{s}\sim d_\mu^{{\pi_{\theta^\prime_*}}}(\cdot)}\E_{\bar{a}^{\mathcal{P}^i}\sim\overline{\pi}^{\mathcal{P}^i}_{\theta_*}(\cdot|\bar{s})}\E_{\bar{a}^{i}\sim\pi^{i}_{\tilde{\theta}^i_*}(\cdot|\bar{s},\bar{a}^{\mathcal{P}^i})}\E_{\bar{a}^{-\mathcal{P}^i_+}\sim\pi^{-\mathcal{P}^i_+}_{\theta_*}(\cdot|\bar{s},a^{\mathcal{P}^i_+})}\Big[Q^{\pi_{\theta_*}}(\bar{s},\bar{a}^{\mathcal{P}^i},\bar{a}^i,\bar{a}^{-\mathcal{P}^i_+})-V^{\pi_{\theta_*}}(\bar{s})\Big]$$

By lemma \ref{lemma_15} which proves either $\overline{\pi}^{\mathcal{P}^i}_{\theta_*}(\cdot|\bar{s})=0$ or $I_{+}^{s,a^{\mathcal{P}^i},i}=\emptyset$,
$$\leq\E_{\bar{s}\sim d_\mu^{{\pi_{\theta^\prime_*}}}(\cdot)}\E_{\bar{a}^{\mathcal{P}^i}\sim\pi^{\mathcal{P}^i}_{\theta_*}(\cdot|\bar{s})}\E_{\bar{a}^{i}\sim\pi^{i}_{\tilde{\theta}^i_*}(\cdot|\bar{s},\bar{a}^{\mathcal{P}^i})}\E_{\bar{a}^{-\mathcal{P}^i}\sim\pi^{-\mathcal{P}^i}_{\theta_*}(\cdot|\bar{s},a^{\mathcal{P}^i})}\Big[Q^{\pi_{\theta_*}}(\bar{s},\bar{a}^{\mathcal{P}^i},\bar{a}^{-\mathcal{P}^i})-V^{\pi_{\theta_*}}(\bar{s})\Big]$$

$$=\E_{\bar{s}\sim d_\mu^{{\pi_{\theta^\prime_*}}}(\cdot)}\E_{\bar{a}^{\mathcal{P}^i}\sim\pi^{\mathcal{P}^i}_{\theta_*}(\cdot|\bar{s})}\E_{\bar{a}^{-\mathcal{P}^i}\sim\pi^{-\mathcal{P}^i}_{\theta_*}(\cdot|\bar{s},a^{\mathcal{P}^i})}\Big[Q^{\pi_{\theta_*}}(\bar{s},\bar{a}^{\mathcal{P}^i},\bar{a}^{-\mathcal{P}^i})-V^{\pi_{\theta_*}}(\bar{s})\Big]$$
$$=\E_{\bar{s}\sim d_\mu^{{\pi_{\theta^\prime_*}}}(\cdot)}\Big[V^{\pi_{\theta_*}}(\bar{s})-V^{\pi_{\theta_*}}(\bar{s})\Big]=0$$.
$$\longrightarrow V^{\pi_{\theta^\prime_*}}\leq V^{\pi_{\theta_*}}$$
Therefore, $(\pi_{\theta^1_*}^1, \cdots, \pi_{\theta^N_*}^N,\mathcal{G})$ is a Nash policy.
\end{proof}

\begin{corollary}[Asymptotic convergence of BN policy gradient to optimal fully-correlated BN joint policy]
\label{proof to corollary 1}
Under Assumptions \ref{assumption:discounted state visitation distribution} - \ref{assumption:PG} 
and additional Assumption \ref{assumption:augumented_state_visitation} that assumes positive visitation measure for any augmented state, suppose every agent $i\in\mathcal{N}$ follows the policy gradient dynamics \eqref{eq:PG}, which results in the update dynamics \eqref{eq:BayesianPG} for each each agent $i$, parent actions $a^{\mathcal{P}^i}$, and local action $a^i$, with $\eta\leq \frac{(1-\gamma)^3}{8N(r_{\rm max}-r_{\rm min})}$, then the converged fully-correlated BN policy 
$(\pi_{\theta^1_*}^1, \cdots, \pi_{\theta^N_*}^N,\mathcal{G})$ is an optimal policy.
\end{corollary}

\begin{proof}
We can assume without loss of generality that agents $1\cdots N$ in $\mathcal{G}$ has a topological ordering of $1\cdots N$ (This means that agent $i\in\mathcal{N}$ is the source of $N-i$ edges and target of $i-1$ edges). 
\\Note that in this case, $\forall i, a^{\mathcal{P}^i}=[a^{\mathcal{P}^{i-1}_+}]$,
\begin{align}
Q^{\pi_\theta,i}(s,a^{\mathcal{P}^i}) &= \E_{\bar{a}^{-\mathcal{P}^i}\sim\pi_{\theta}(\cdot|s,a^{\mathcal{P}^i})}\Big[Q^{\pi_\theta}(s,a^{\mathcal{P}^i},\bar{a}^{-\mathcal{P}^i})\Big] \nonumber\\&= \E_{\bar{a}^{-\mathcal{P}^{i-1}_+}\sim\pi_{\theta}(\cdot|s,a^{\mathcal{P}^{i-1}_+})}\Big[Q^{\pi_\theta}(s,a^{\mathcal{P}^{i-1}_+},\bar{a}^{-\mathcal{P}^{i-1}_+})\Big]=Q^{\pi_\theta,i-1}(s,a^{\mathcal{P}^{i-1}_+}) \label{eq:corollary}
\end{align}
\\With assumption \ref{assumption:augumented_state_visitation}, we know that $\forall i,a^{\mathcal{P}^i},I_{+}^{s,a^{\mathcal{P}^i},i}=\emptyset$.
\\$\forall a=[a^{\mathcal{P}^{N}},a^N]$,
$$Q^{\pi_{\theta_*}}(s,a)=Q^{\pi_{\theta_*}}(s,a^{\mathcal{P}^{N}},a^N)=Q^{\pi_{\theta_*},N}(s,a^{\mathcal{P}^{N}},a^N)$$
By $I_{+}^{s,a^{\mathcal{P}^{N},N}}=\emptyset$,
$$\leq Q^{\pi_{\theta_*},N}(s,a^{\mathcal{P}^{N}})$$
By Equation (\ref{eq:corollary}),
$$= Q^{\pi_{\theta_*},N-1}(s,a^{\mathcal{P}^{N-1}},a^{N-1})$$
By $I_{+}^{s,a^{\mathcal{P}^{N-1},N-1}}=\emptyset$,
$$\leq Q^{\pi_{\theta_*},N-1}(s,a^{\mathcal{P}^{N-1}})$$
By Equation (\ref{eq:corollary}),
$$= Q^{\pi_{\theta_*},N-2}(s,a^{\mathcal{P}^{N-2})},a^{N-2})$$
By keep doing the same procedure above,
$$\leq Q^{\pi_{\theta_*},1}(s,a^{\mathcal{P}^{1}})$$
Since $a^{\mathcal{P}^{1}}=\emptyset$,
$$=V^{\pi_{\theta_*}}(s)$$
Then, since $\forall s,a, Q^{\pi_{\theta_*}}(s,a)\leq V^{\pi_{\theta_*}}(s)$, we know that $(\pi_{\theta^1_*}^1, \cdots, \pi_{\theta^N_*}^N,\mathcal{G})$ is an optimal policy.
\end{proof}

\newpage\clearpage
\section{Experiment details}
\label{Experiment details}
\subsection{Tabular Coordination Game}
\subsubsection{Pseudocode for the reward function in Coordination Game}
\begin{algorithm}
	\caption{Calculate the team reward for $N$ agents in state $s$}
	\begin{algorithmic}
	    \IF {$(N=2)$ or $(N=3)$}
		\STATE \texttt{difference\_bound=1}
        \ELSE 
        \STATE \texttt{difference\_bound=2}
		\ENDIF
		
		\IF {$\text{abs}(s.\text{count(0)}-s.\text{count(1))}\leq \text{difference\_bound}$}
		\IF {$s.\text{count(0)}<s.\text{count(1)}$}
		\STATE reward$=1$
		\ELSE 
		\STATE reward$=0$
		\ENDIF

		\ELSIF {$s.\text{count(0)}>s.\text{count(1)}$}
		\STATE reward$=3$
        \ELSE 
        \STATE reward$=2$
		\ENDIF
	
	\end{algorithmic} 
\end{algorithm}

\subsubsection{Coordination Game Environment Hyperparameters}
\begin{table}[ht]
\centering
\caption{CG Env Hyperparameters}
\label{table:Env Hyperparameters}
\begin{tabular}{ll} 
\toprule
Hyperparameter                     & Value                                          \\ 
\hline
$\gamma$ (discount factor)         & 0.95                                           \\
$\mu$ (initial state distribution) & Uniform                                        \\
$\epsilon$                    & 0.1  \\
\bottomrule
\end{tabular}
\end{table}

\subsubsection{Hyperparameters for Coordination Game (tabular)}
\begin{table}[ht]
\centering
\caption{Hyperparameters}
\label{table:Implementation details}
\begin{tabular}{ll} 
\toprule
Hyperparameter                     & Value                                          \\ 
\hline
Environment steps         & 2e5, 1e6, 2e7 for CG,Aloha, and SMAC, respectively.                                           \\
Episode length & 20, 25, 400 for CG,Aloha, and SMAC, respectively.                                        \\
PPO epoch             & 5 for all environments.                                          \\
Critic Learning rate                             & 7e-4 for CG and aloha, and 5e-4 for SMAC.                                \\
Actor Learning rate                                  & 7e-4 for CG and aloha, and 5e-4 for SMAC.
\\Optimizer                                  & Adam.
\\
\#Episodes for evaluation & 100 for CG and aloha, and 32 for SMAC.
\\
\#Rollout threads & 32 for CG and aloha, 8 for SMAC.
\\
\#Training threads & 32 for CG and aloha, 1 for SMAC.
\\
Hidden size & 64 for all environments.
\\
Actor architecture for CG                   & Concat(\text{Base}$(s),a^{\mathcal{P}^i})$)-FC(action dim)-softmax  \\
Actor architecture for Aloha                   & Concat(\text{Base}$(o^i),\text{Base}(\text{Concat}(o^{\mathcal{P}^i},a^{\mathcal{P}^i})$)-FC(action dim)-softmax  \\
Actor architecture for SMAC                   & Concat($o^i,a^{\mathcal{P}^i}$)-Base(hidden)-softmax  \\
Edge Net architecture for CG                    & Concat($\{o^i\}_{i=1}^N$)-DeepSet\\
Edge Net architecture for Aloha and SMAC                    & Concat($\{o^i\}_{i=1}^N$)-FC(hidden)-Relu-FC($2N^2$)\\
Permutation Net architecture for CG and Aloha                   & Concat($\{o^i\}_{i=1}^N$)-FC(hidden)-Relu-FC($2N^2$)\\
Permutation Net architecture for SMAC                   & Always output identity matrix\\
Critic architecture for all environments                  & joint observation or state-Base(hidden)-FC(1)  \\
\bottomrule
\end{tabular}

Coordination Game is abbreviated as CG.
\\Base(hidden): FC(hidden)-Relu-FC(hidden)-Relu 
\\DeepSetEncoder: FC(hidden)-Relu-FC(hidden)-Relu-FC(hidden)
\\DeepSetDecoder: FC(hidden)-Relu-FC(hidden)-Relu-FC($2N^2$)
\\DeepSet: input-DeepSetEncoder-mean(dim for agents)-DeepSetEncoder
\end{table}


\end{document}